\documentclass[draft]{llncs}
\usepackage{mathpartir}
\usepackage{amsmath,amssymb,mathtools,xspace,upgreek,url}
\usepackage{float}

\floatstyle{boxed} 
\restylefloat{figure}

\usepackage{nfBisim}

\title{Normal Form Bisimulations\\ for Delimited-Control Operators}

\author{Dariusz Biernacki \and Sergue\"i Lenglet}

\institute{University of Wroc\l{}aw}

\begin{document}

\pagestyle{plain}

\maketitle

\thispagestyle{plain}

\begin{abstract}
  We define a notion of normal form bisimilarity for the untyped
  call-by-value $\lambda$-calculus extended with the delimited-control
  operators shift and reset. Normal form bisimilarities are simple,
  easy-to-use behavioral equivalences which relate terms without
  having to test them within all contexts (like contextual
  equivalence), or by applying them to function arguments (like
  applicative bisimilarity). We prove that the normal form
  bisimilarity for shift and reset is sound but not complete
  w.r.t. contextual equivalence and we define up-to techniques that
  aim at simplifying bisimulation proofs. Finally, we illustrate the
  simplicity of the techniques we develop by proving several
  equivalences on terms.
\end{abstract}

\section{Introduction}

Morris-style contextual equivalence \cite{JHMorris:PhD} is usually
considered as the most natural behavioral equivalence for functional
languages based on $\lambda$-calculi. Roughly, two terms are
equivalent if we can exchange one for the other in a bigger program
without affecting its behavior (\ie whether it terminates or not). The
quantification over program contexts makes contextual equivalence hard
to use in practice and, therefore, it is common to look for
easier-to-use behavioral equivalences, such as \emph{bisimilarities}.

Several kinds of bisimilarity relations have been defined so far, such
as \emph{applicative bisimilarity} \cite{Abramsky-Ong:IaC93},
\emph{normal form bisimilarity} \cite{Lassen:LICS05} (originally
defined in \cite{Sangiorgi:LICS92}, where it was called \emph{open
  bisimilarity}), and \emph{environmental bisimilarity}
\cite{Sangiorgi-al:LICS07}. Applicative and environmental
bisimilarities usually compare terms by applying them to function
arguments; as a result, we obtain relations which completely
characterize contextual equivalence, but still contain a universal
quantification over arguments in their definitions. In contrast,
normal form bisimilarity does not need such quantification; it equates
terms by reducing them to normal form, and by requiring the sub-terms
of these normal forms to be bisimilar. Normal form relations are
convenient in practice, but they are usually not \emph{complete}
w.r.t. contextual equivalence, \ie there exist contextually equivalent
terms that are not normal form bisimilar.

A notion of normal form bisimulation has been defined in various
calculi, including the pure $\lambda$-calculus
\cite{Lassen:MFPS99,Lassen:LICS05}, the $\lambda$-calculus with
ambiguous choice \cite{Lassen:MFPS05}, the $\lambda\mu$-calculus
\cite{Lassen:LICS06}, and the $\lambda\mu\rho$-calculus
\cite{Stoevring-Lassen:POPL07}, where normal form bisimilarity
completely characterizes contextual equivalence. However, it has not
yet been defined for calculi with \emph{delimited-control} operators,
such as shift and reset~\cite{Danvy-Filinski:LFP90}---programming
constructs rapidly gaining currency in the recent years. Unlike
abortive control operators (such as call/cc), delimited-control
operators allow to delimit access to the current continuation and to
compose continuations. The operators shift and reset were introduced
as a direct-style realization of the traditional success/failure
continuation model of backtracking otherwise expressible only in
continuation-passing style~\cite{Danvy-Filinski:LFP90}. The numerous
theoretical and practical applications of shift and reset (see,
e.g.,~\cite{Biernacka-al:LMCS05} for an extensive list) include the
seminal result by Filinski showing that a programming language endowed
with shift and reset is monadically complete~\cite{Filinski:POPL94}.

Up to now, only an applicative bisimilarity has been defined for a
calculus with shift and reset \cite{Biernacki-Lenglet:FOSSACS12}. In
this paper, we define several notions of normal form bisimilarity for
such a calculus, more tractable than contextual equivalence or
applicative bisimilarity. We prove they are \emph{sound}
w.r.t. contextual equivalence (\ie included in contextual
equivalence), but fail to be complete. We also develop \emph{up-to
  techniques} that are helpful when proving equivalences with normal
form bisimulations.

In Section \ref{s:calculus}, we define the $\lambda$-calculus with
delimited control that we use in this paper, and we recall the
definition of contextual equivalence of
\cite{Biernacki-Lenglet:FOSSACS12} for this calculus. We then define
in Section \ref{s:nf-bisim} the main notion of normal form
bisimilarity and we prove its properties. In Section
\ref{s:refined-up-to}, we refine the definition of normal form
bisimilarity to relate more contextually equivalent terms, at the cost
of extra complexity in bisimulation proofs. We also propose several
up-to techniques which simplify the proofs of equivalence of terms. In
Section~\ref{s:examples}, we illustrate the simplicity of use
(compared to applicative bisimilarity) of the notions we define by
employing them in the proofs of several equivalences of terms.
Section~\ref{s:conclusion} concludes the paper, and Appendix
\ref{a:soundness} contains the congruence proofs of the considered
normal form bisimilarities.

\section{The Calculus $\lamshift$}
\label{s:calculus}

In this section, we present the syntax, reduction semantics, and
contextual equivalence for the language $\lamshift$ studied throughout
this article.

\subsection{Syntax}

The language $\lamshift$ extends the call-by-value $\lambda$-calculus
with the delimited-control operators \emph{shift} and
\emph{reset}~\cite{Danvy-Filinski:LFP90}. We assume we have a set of
term variables, ranged over by $\varx$, $y$, $z$, and $\vark$. We use
the metavariable $k$ for term variables representing a continuation
(\eg when bound with a shift), while $x$, $y$, and $z$ stand for any
values; we believe such distinction helps to understand examples and
reduction rules. The syntax of terms and values is given by the
following grammars:
$$
\begin{array}{llll}
  \textrm{Terms:} & \tm & ::= & \varx \Mid \lam \varx \tm \Mid \app \tm \tm \Mid
  \shift \vark \tm \Mid \reset \tm \\
  \textrm{Values:}& \val & ::= & \lam \varx \tm \Mid \varx
\end{array}
$$ The operator \emph{shift} ($\shift \vark \tm$) is a capture
operator, the extent of which is determined by the delimiter
\emph{reset} ($\rawreset$). A $\lambda$-abstraction $\lam \varx \tm$
binds $\varx$ in $\tm$ and a shift construct $\shift \vark \tm$ binds
$\vark$ in $\tm$; terms are equated up to $\alpha$-conversion of their
bound variables. The set of free variables of $\tm$ is written $\fv
\tm$; a term is \emph{closed} if it does not contain free variables.

We distinguish several kinds of contexts, as follows. 
$$
\begin{array}{llll}
  \textrm{Pure contexts:}& \ctx & ::= & \mtctx \Mid \vctx \val \ctx \Mid \apctx \ctx
  \tm \\
  \textrm{Evaluation contexts:\:}& \rctx & ::= & \mtctx \Mid \vctx \val \rctx \Mid \apctx \rctx
  \tm \Mid \resetctx \rctx \\
  \textrm{Contexts:}& \cctx & ::= & \mtctx \Mid \lam \varx \cctx
  \Mid \vctx \tm \cctx \Mid \apctx \cctx \tm \Mid \shift \vark \cctx \Mid \resetctx
  \cctx \\
\end{array}
$$ Regular contexts are ranged over by $\cctx$. The pure evaluation
contexts\footnote{This terminology comes from Kameyama (\eg in
  \cite{Kameyama-Hasegawa:ICFP03}).} (abbreviated as pure contexts),
ranged over by $\ctx$, represent delimited continuations and can be
captured by the shift operator. The call-by-value evaluation contexts,
ranged over by $\rctx$, represent arbitrary continuations and encode
the chosen reduction strategy. Filling a context $\cctx$ (respectively
$\ctx$, $\rctx$) with a term $\tm$ produces a term, written $\inctx
\cctx \tm$ (respectively $\inctx \ctx \tm$, $\inctx \rctx \tm$); the
free variables of $\tm$ may be captured in the process. A context is
\emph{closed} if it contains only closed terms.

\subsection{Reduction Semantics}
\label{ss:reduction}

Before we present the reduction semantics for $\lamshift$, let us
briefly describe an intuitive semantics of shift and reset by means of
an example written in SML, using Filinski's implementation of shift
and reset~\cite{Filinski:POPL94}.
\begin{example}
  \label{e:copy}
The following function copies a list~\cite{Biernacki-al:RS-05-16},
where the SML expression {\tt shift (fn k => t)} corresponds to
$\shift k t$ and {\tt reset (fn () => t)} corresponds to $\reset t$:
\begin{quote}
\begin{verbatim}
fun copy xs = 
  let fun visit nil = nil
        | visit (x::xs) = visit (shift (fn k => x :: (k xs)))
  in reset (fn () => visit xs) end
\end{verbatim}
\end{quote}

This simple function illustrates the main ideas of programming with
shift and reset:
\begin{itemize}
\item[$\bullet$] The control delimiter reset delimits
  continuations. Any control effects occurring in the subsequent
  calls to function {\tt visit} are local to function {\tt copy}.
\item[$\bullet$] The control operator shift captures delimited
  continuations. Each but last recursive call to {\tt visit} abstracts
  the continuation that can be represented as a function {\tt fn v =>
    reset (fn () => visit v)} and binds it to {\tt k}.
\item[$\bullet$] Captured continuations are composed statically. When
  applied, in the expression {\tt x :: (k xs)}, the captured
  continuation becomes the current delimited continuation that is
  isolated from the rest of the program, and in particular from the
  expression {\tt x ::}, by a control delimiter---witness the control
  delimiter in the expression {\tt fn v => reset (fn () => visit v)}
  representing the captured continuation.
\end{itemize}
\end{example}

Formally, the call-by-value reduction semantics of $\lamshift$ is defined as
follows, where $\subst \tm \varx \val$ is the usual capture-avoiding
substitution of $\val$ for $\varx$ in $\tm$:
$$
\begin{array}{lrll}
  \RRbeta & \quad \inctx \rctx {\app {\lamp \varx \tm} \val} & \redcbv & \inctx
  \rctx {\subst \tm \varx \val} \\
  \RRshift & \quad \inctx \rctx {\reset{\inctx \ctx {\shift \vark \tm}}} &
  \redcbv & \inctx \rctx{\reset{\subst \tm \vark
    {\lam \varx {\reset {\inctx \ctx \varx}}}}} \mbox{ with } \varx \notin \fv
  \ctx\\
  \RRreset & \quad \inctx \rctx {\reset \val} & \redcbv & \inctx \rctx \val
\end{array}
$$ 
The term $\app {\lamp \varx \tm} \val$ is the usual call-by-value redex for
$\beta$-reduction (rule $\RRbeta$). The operator $\shift \vark \tm$ captures its
surrounding context $\ctx$ up to the dynamically nearest enclosing reset, and
substitutes $\lam \varx {\reset {\inctx \ctx \varx}}$ for $\vark$ in $\tm$ (rule
$\RRshift$). If a reset is enclosing a value, then it has no purpose as a
delimiter for a potential capture, and it can be safely removed (rule
$\RRreset$). All these reductions may occur within a metalevel context
$\rctx$. The chosen call-by-value evaluation strategy is encoded in the grammar
of the evaluation contexts.

\begin{example}
  \label{e:reduction}
  Let $i = \lam \varx \varx$ and $\omega = \lam \varx {\app \varx \varx}$. We
  present the sequence of reductions initiated by $\reset {\app{\appp {(\shift
        {\vark_1}{\app i {\appp {\vark_1} i}})}{\shift {\vark_2} \omega}}{\appp
      \omega \omega}}$. The term $\shift {\vark_1}{\app i {\appp {\vark_1} i}}$
  is within the pure context $\ctx = \apctx {(\apctx \mtctx {\shift {\vark_2} \omega})} {\appp \omega
      \omega}$, enclosed in a delimiter $\rawreset$, so $\ctx$ is captured
  according to rule $\RRshift$.
  $$\reset {\app{\appp {(\shift{\vark_1}{\app i {\appp {\vark_1}
            i}})}{\shift {\vark_2} \omega}}{\appp \omega \omega}}
  \redcbv \reset{\app i {\appp {\lamp \varx {\reset {\app{\appp \varx
              {\shift {\vark_2} \omega}}{\appp \omega \omega}}}} i}}$$
  The role of reset in $\lam \varx {\reset{\inctx \ctx \varx}}$ is
  more clear after reduction of the $\beta_v$-redex $\app {\lamp \varx
    {\reset{\inctx \ctx \varx}}} i$.
  $$\reset{\app i {\appp {\lamp \varx {\reset {\app{\appp \varx {\shift {\vark_2}
                \omega}}{\appp \omega \omega}}}} i}} \redcbv \reset {\app i
    {\reset{\app {\appp i {\shift {\vark_2} \omega}}{\appp \omega \omega}}}}$$
  When the captured context $\ctx$ is reactivated, it is not \emph{merged} with
  the context $\vctx i \mtctx$, but \emph{composed} thanks to the
  reset enclosing $\ctx$. As a result, the capture triggered by $\shift
  {\vark_2} \omega$ leaves the term $i$ outside the first enclosing reset
  untouched. 
  $$\reset {\app i {\reset{\app {\appp i {\shift {\vark_2}
            \omega}}{\appp \omega \omega}}}} \redcbv \reset {\app i
    {\reset \omega}}$$ Because $\vark_2$ does not occur in $\omega$,
  the context $\apctx {(\vctx i \mtctx)} {\appp \omega \omega}$ is
  discarded when captured by $\shift {\vark_2} \omega$. Finally, we
  remove the useless delimiter $\reset {\app i {\reset \omega}}
  \redcbv \reset{\app i \omega}$ with rule $\RRreset$, and we then
  $\beta_v$-reduce and remove the last delimiter $\reset{\app i
    \omega} \redcbv \reset{\omega} \redcbv \omega$. Note that while
  the reduction strategy is call-by-value, some function arguments are
  not evaluated, like the non-terminating term $\app \omega \omega$
  in this example.  
\end{example}

There exist terms which are not values and which cannot be reduced any
further; these are called \emph{stuck terms}.
\begin{definition}
  A term $\tm$ is stuck if $\tm$ is not a value and $\tm \not
  \redcbv$.
\end{definition}
For example, the term $\inctx \ctx {\shift \vark \tm}$ is stuck
because there is no enclosing reset; the capture of $\ctx$ by the
shift operator cannot be triggered. In fact, stuck terms are easy to
characterize.
\begin{lemma}
  \label{l:stuck}
  A term $\tm$ is stuck iff $\tm = \inctx \ctx {\shift \vark {\tm'}}$
  for some $\ctx$, $k$, and $\tm'$ or $\tm = \inctx \rctx {\app \varx
    \val}$ for some $\rctx$, $\varx$, and $\val$.
\end{lemma}
We call \emph{control stuck terms} terms of the form $\inctx \ctx
{\shift \vark \tm}$ and \emph{open stuck terms} the terms of the form
$\inctx \rctx {\app \varx \val}$.
\begin{definition}
  A term $\tm$ is a normal form, if $\tm$ is a value or a stuck term.
\end{definition}

We call \emph{redexes} (ranged over by $\redex$) terms of the form
$\app{\lamp \varx \tm} \val$, $\reset {\inctx \ctx {\shift \vark
    \tm}}$, and $\reset \val$. Thanks to the following
unique-decomposition property, the reduction relation $\redcbv$ is
deterministic.
\begin{lemma}
  For all terms $\tm$, either $\tm$ is a normal form, or there exist a
  unique redex $\redex$ and a unique context $\rctx$ such that $\tm =
  \inctx \rctx \redex$.
\end{lemma}
 
Finally, we write $\clocbv$ for the transitive and reflexive closure
of $\redcbv$, and we define the evaluation relation of $\lamshift$ as
follows.
\begin{definition}
  We write $\tm \evalcbv \tm'$ if $\tm \clocbv \tm'$ and $\tm' \not\redcbv$. 
\end{definition}
The result of the evaluation of a term, if it exists, is a normal
form. If a term $\tm$ admits an infinite reduction sequence, we say it
\emph{diverges}, written $\tm \divcbv$. In the rest of the article, we
use extensively $\Omega = \app {\lamp \varx {\app \varx \varx}}{\lamp
  \varx {\app \varx \varx}}$ as an example of such a term.

\subsection{Contextual Equivalence}

In this paper, we use the same contextual equivalence as in
\cite{Biernacki-Lenglet:FOSSACS12}, where control stuck terms can be
observed. Note that this relation is a bit more discriminative than
simply observing termination, as pointed out in
\cite{Biernacki-Lenglet:FOSSACS12}.
\begin{definition}
  Let $\tmzero$, $\tmone$ be terms. We write $\tmzero \ctxequiv
  \tmone$ if for all $\cctx$ such that $\inctx \cctx \tmzero$ and
  $\inctx \cctx \tmone$ are closed, the following hold:
  \begin{itemize}
  \item $\inctx \cctx \tmzero \evalcbv \valzero$ implies $\inctx
    \cctx \tmone \evalcbv \valone$;
  \item $\inctx \cctx \tmzero \evalcbv \tmzero'$, where $\tmzero'$
    is control stuck, implies $\inctx \cctx \tmone \evalcbv \tmone'$, with
    $\tmone'$ control stuck as well;
  \end{itemize}
  and conversely for $\inctx \cctx \tmone$.
\end{definition}
We can simplify the proofs of contextual equivalence of terms by
relying on the following context lemma~\cite{Milner:TCS77} for
$\lamshift$ (for a proof see Definition 5 and Section~3.4
in~\cite{Biernacki-Lenglet:FOSSACS12}). Instead of testing terms with
(free-variables capturing) general contexts, we can simply first close
them (using closed values) and then put them within (closed)
evaluation contexts.

\begin{lemma}[Context Lemma]
  \label{l:context-lemma}
  We have $\tmzero \ctxequiv \tmone$ iff for all closed contexts $\rctx$ and for
  all substitutions $\sigma$ (mapping variables to closed values) such that
  $\tmzero\sigma$ and $\tmone\sigma$ are closed, the following hold:
  \begin{itemize}
  \item $\inctx \rctx {\tmzero\sigma} \evalcbv \valzero$ implies
    $\inctx \rctx {\tmone\sigma} \evalcbv \valone$;
  \item $\inctx \rctx {\tmzero\sigma} \evalcbv \tmzero'$, where
    $\tmzero'$ is control stuck, implies $\inctx \rctx {\tmone\sigma} \evalcbv
    \tmone'$, with $\tmone'$ control stuck as well;
  \end{itemize}
  and conversely for $\inctx \rctx {\tmone\sigma}$.
\end{lemma}
In the rest of the paper, when proving that terms are contextually equivalent,
we implicitly use Lemma \ref{l:context-lemma}.

\section{Normal Form Bisimilarity}
\label{s:nf-bisim}

In this section, we discuss a notion of bisimulation based on the
evaluation of terms to normal forms. The difficulties are mainly in
the handling of control stuck terms and in the definition of the
relation on non-pure evaluation contexts. We propose here a first way
to deal with control stuck terms, that will be refined in the next
section. In any definitions or proofs, we say a variable is
\emph{fresh} if it does not occur free in the terms or contexts under
consideration.

\subsection{Definition}

Following Lassen's approach \cite{Lassen:LICS05}, we define a normal
form bisimulation where we relate terms by comparing the results of
their evaluation (if they exist). As we need to compare terms as well
as evaluation contexts, we extend a relation $\rel$ on terms to
contexts in the following way: we write $\rctxzero \rel \rctxone$ if
$\rctxzero = \inctx{\rctxzero'}{\reset \ctxzero}$, $\rctxone =
\inctx{\rctxone'}{\reset \ctxone}$, $\inctx {\rctxzero'} \varx \rel
\inctx {\rctxone'} \varx$, and $\reset{\inctx {\ctxzero} \varx} \rel
\reset{\inctx {\ctxone} \varx}$ for a fresh $\varx$, or if $\rctxzero =
\ctxzero$, $\rctxone = \ctxone$, and $\inctx \ctxzero \varx \rel
\inctx \ctxone \varx$ for a fresh $\varx$. The rationale behind this
definition is explained later. Following \cite{Lassen:LICS05}, we
define the application $\appval \val y$ as $\app \varx y$ if $\val =
\varx$, and as $\subst \tm \varx y$ if $\val = \lam \varx
\tm$. Finally, given a relation $\rel$ on terms, we write $\inv\rel$
for its inverse, and we inductively define a relation $\nf\rel$ on
normal forms as follows:
\begin{mathpar}
  \inferrule{\appval \valzero \varx \rel \appval \valone \varx \\ \varx \textrm{
      fresh}}
  {\valzero \nf\rel \valone}
  \and
  \hspace{-1em}\inferrule{\ctxzero \rel \ctxone \\ \reset \tmzero \rel \reset \tmone}
  {\inctx \ctxzero {\shift \vark \tmzero} \nf\rel \inctx \ctxone {\shift \vark
      \tmone}}
  \and
  \hspace{-1em}\inferrule{\rctxzero \rel \rctxone \\ \valzero \nf\rel \valone}
  {\inctx \rctxzero {\app \varx \valzero} \nf\rel \inctx \rctxone {\app \varx
      \valone}}
\end{mathpar}

\begin{definition}
  \label{d:nf-bisim}
  A relation $\rel$ on terms is a normal form simulation if $\tmzero \rel
  \tmone$ and $\tmzero \evalcbv \tmzero'$ implies $\tmone \evalcbv \tmone'$ and
  $\tmzero' \nf\rel \tmone'$. A relation $\rel$ is a normal form bisimulation if
  both $\rel$ and $\inv\rel$ are normal form simulations. Normal form
  bisimilarity, written $\bisim$, is the largest normal form bisimulation.
\end{definition}

Henceforth, we often drop the ``normal form'' attribute when talking about
bisimulations for brevity. Two terms $\tmzero$ and $\tmone$ are bisimilar if
their evaluations lead to matching normal forms (\eg if $\tmzero$ evaluates to a
control stuck term, then so does $\tmone$) with bisimilar sub-components. We now
detail the different cases.

Normal form bisimilarity does not distinguish between evaluation to a variable
and evaluation to a $\lambda$-abstraction. Instead, we relate terms evaluating
to any values $\valzero$ and $\valone$ by comparing $\appval \valzero \varx$ and
$\appval \valone \varx$, where $\varx$ is fresh. As originally pointed out by
Lassen \cite{Lassen:LICS05}, this is necessary for the bisimilarity to be sound
w.r.t. $\eta$-expansion; otherwise it would distinguish $\eta$-equivalent terms
such as $\lam y {\app \varx y}$ and $\varx$. Using $\rawappval$ instead of
regular application avoids the introduction of unnecessary $\beta$-redexes,
which could reveal themselves problematic in proofs.

For a control stuck term $\inctx \ctxzero {\shift \vark {\tmzero}}$ to
be executed, it has to be plugged into an evaluation context
surrounded by a reset; by doing so, we obtain a term of the form
$\reset{\subst {\tmzero} \vark {\lam \varx {\reset {\inctx {\ctxzero'}
        \varx}}}}$ for some context $\ctxzero'$. Notice that the
resulting term is within a reset; similarly, when comparing $\inctx
\ctxzero {\shift \vark {\tmzero}}$ and $\inctx \ctxone {\shift \vark
  {\tmone}}$, we ask for the shift bodies $\tmzero$ and $\tmone$ to be
related when surrounded by a reset. We also compare $\ctxzero$ and
$\ctxone$, which amounts to executing $\inctx \ctxzero \varx$ and
$\inctx \ctxone \varx$ for a fresh~$\varx$, since the two contexts are
pure. Comparing $\tmzero'$ and $\tmone'$ without reset would be too
discriminating, as it would distinguish the two contextually
equivalent terms $\shift \vark {\reset \tm}$ and $\shift \vark
\tm$.\footnote{The equivalence $\shift \vark {\reset \tm} \equiv
  \shift \vark \tm$ comes from Kameyama and Hasegawa's axiomatization
  of shift and reset \cite{Kameyama-Hasegawa:ICFP03} and has been
  proved using applicative bisimilarity
  in~\cite{Biernacki-Lenglet:FOSSACS12}.} Indeed, without reset, we
would have to relate $\reset \tm$ and $\tm$, which are not equivalent
in general (take $\tm=\shift {\vark'} \val$ for some $\val$), while
Definition~\ref{d:nf-bisim} requires $\reset {\reset \tm}$ and $\reset
\tm$ to be related (which holds for all $\tm$; see
Example~\ref{e:reset-reset}).

Two normal forms $\inctx \rctxzero {\app \varx \valzero}$ and $\inctx \rctxone
{\app \varx \valone}$ are bisimilar if the values $\valzero$ and $\valone$ as
well as the contexts $\rctxzero$ and $\rctxone$ are related. We have to be
careful when defining bisimilarity on (possibly non pure) evaluation
contexts. We cannot simply relate $\rctxzero$ and $\rctxone$ by executing
$\inctx \rctxzero y$ and $\inctx \rctxone y$ for a fresh $y$. Such a definition
would equate the contexts $\mtctx$ and $\reset \mtctx$, which in turn would
relate the terms $\app \varx \val$ and $\reset {\app \varx \val}$, which are not
contextually equivalent: they are distinguished by the context $\app {\lamp
  \varx \mtctx}{\lam y {\shift \vark \Omega}}$. A context containing a reset
enclosing the hole should be related only to contexts with the same
property. However, we do not want to precisely count the number of delimiters
around the hole; doing so would distinguish $\reset \mtctx$ and $\reset {\reset
  \mtctx}$, and therefore it would discriminate the contextually equivalent
terms $\reset{\app \varx \val}$ and $\reset{\reset {\app \varx \val}}$. Hence,
the extension of bisimulation to contexts (given before Definition
\ref{d:nf-bisim}) checks that if one of the contexts contains a reset
surrounding the hole, then so does the other; then it compares the contexts
beyond the first enclosing delimiter by simply evaluating them using a fresh
variable.  As a result, it rightfully distinguishes $\mtctx$ and $\reset
\mtctx$, but it relates $\reset \mtctx$ and $\reset {\reset \mtctx}$.

\begin{example}
  \label{e:reset-reset}
  We prove that $\reset \tm \bisim \reset{\reset \tm}$ by showing that
  $\rel = \{ (\reset \tm, \reset {\reset \tm}) \} \cup \bisim$ is a
  bisimulation. If $\reset \tm \evalcbv \val$, then $\reset {\reset
    \tm} \evalcbv \val$, and $\val \nf\bisim \val$ holds. The case
  $\reset \tm \evalcbv \inctx \ctx {\shift \vark {\tm'}}$ is not
  possible; one can check that if $\reset \tm \redcbv \tm'$, then
  $\tm'$ is a value, or can be written $\reset{\tm''}$ for some
  $\tm''$ (and the same holds for $\reset \tm \evalcbv \tm'$).

  If $\reset \tm \evalcbv \inctx \rctx {\app \varx \val}$, then there
  exists $\rctx'$ such that $\tm \evalcbv \inctx {\rctx'}{\app \varx
    \val}$ and $\rctx = \reset {\rctx'}$. Therefore, we have $\reset
  {\reset \tm} \evalcbv \reset{\reset{\inctx {\rctx'}{\app \varx
        \val}}}$. We have $\val \nf\bisim \val$, and we have to prove
  that $\reset {\rctx'} \rel \reset{\reset {\rctx'}}$ to conclude. If
  $\rctx'$ is a pure context $\ctx$, then we have to prove
  $\reset{\inctx \ctx y} \rel \reset{\inctx \ctx y}$ and $y \rel
  \reset y$ for a fresh $y$, which are both true because $\bisim
  \subseteq \rel$. If $\rctx' = \inctx {\rctx''}{\reset {\ctx}}$, then
  given a fresh $y$, we have to prove $\reset {\inctx {\rctx''} y}
  \rel \reset{\reset {\inctx {\rctx''} y}}$ (clear by the definition
  of $\rel$), and $\reset{\inctx \ctx y} \rel \reset{\inctx \ctx y}$
  (true because $\bisim \subseteq \rel$).
  
  Similarly, it is easy to check that the evaluations of $\reset{\reset \tm}$ are
  matched by $\reset \tm$.
\end{example}

\begin{example}
  \label{e:fixed-point}
  In \cite{Danvy-Filinski:DIKU89}, the authors propose variants of
  Curry's and Turing's call-by-value fixed point combinators using
  shift and reset. Let $\theta = \lam {\varx y}{\app y {\lamp z {\app
        {\app {\app \varx \varx} y} z}}}$. We prove that Turing's
  combinator $\tmzero = \app \theta \theta$ is bisimilar to its shift
  and reset variant $\tmone = \reset {\app \theta {\shift \vark {\app
        \vark \vark}}}$. We build the candidate relation $\rel$
  incrementally, starting from $(\tmzero, \tmone)$. Evaluating
  $\tmzero$ and $\tmone$, we obtain $\tmzero \evalcbv \lam y {\app y
    {\lamp z {\app {\app {\app \theta \theta} y} z}}}=\valzero$ and
  $\tmone \evalcbv \lam y {\app y {\lamp z {\app {\app {\app {\lamp
              \varx {\reset {\app \theta \varx}}}{\lamp \varx {\reset
                {\app \theta \varx}}}} y} z}}}=\valone$; we have to
  add $(\appval \valzero y, \appval \valone y)$ (for a fresh $y$) to
  $\rel$. To relate these terms, we must add $(\appval {\valzero'} z,
  \appval {\valone'} z)$ and $(z, z)$ for a fresh~$z$ to $\rel$, where
  $\valzero' = \lam z {\app {\app {\app \theta \theta} y} z}$ and
  $\valone' = \lam z {\app {\app {\app {\lamp \varx {\reset {\app
              \theta \varx}}}{\lamp \varx {\reset {\app \theta
              \varx}}}} y} z}$. Evaluating $\appval{\valzero'} z$ and
  $\appval{\valone'} z$, we obtain respectively $\app {\app y
    {\valzero'}} z$ and $\app {\app y {\valone'}} z$; to relate these
  two normal forms, we just need to add $(\app \varx z, \app \varx z)$
  (for a fresh $\varx$) to $\rel$, since we already have $\valzero'
  \nf\rel \valone '$. One can check that the constructed relation
  $\rel$ is a normal form bisimulation.

  In contrast, Curry's combinator $\tmzero' = \lam \varx
  {\app{\delta_\varx}{\delta_\varx}}$, where $\delta_\varx = \lam y {\app \varx
    {\lamp z {\app {\app y y} z}}}$, is not bisimilar to its delimited-control
  variant $\tmone'=\lam \varx {\reset {\app {\delta_\varx}{\shift \vark {\app
          \vark \vark}}}}$. Indeed, evaluating the bodies of the two values, we
  obtain respectively $\app \varx {\lamp z {\app {\app
        {\delta_\varx}{\delta_\varx}} z}}$ and $\reset {\reset {\app \varx
      {\lamp z {\app {\app {\lamp y {\reset{\app {\delta_\varx} y}}}{\lamp y
              {\reset{\app {\delta_\varx} y}}}} z}}}}$, and these open stuck
  terms are not bisimilar, because $\mtctx \not\bisim \reset{\reset \mtctx}$. In
  fact, $\tmzero'$ and $\tmone'$ are distinguished by the context $\app \mtctx
  {\lam \varx {\shift \vark \Omega}}$. Finally, we can prove that the two
  original combinators $\app \theta \theta$ and $\lam \varx
  {\app{\delta_\varx}{\delta_\varx}}$ are bisimilar, using the same bisimulation
  as in \cite{Lassen:LICS05}.
\end{example}

\subsection{Soundness and Completeness}
\label{ss:soundness}

Usual congruence proofs for normal form bisimilarities include direct
proofs, where a context and/or substitutive closure of the
bisimilarity is proved to be itself a
bisimulation~\cite{Lassen:MFPS99,Lassen:MFPS05,Stoevring-Lassen:POPL07},
and proofs based on continuation-passing style (CPS) translations
\cite{Lassen:LICS05,Lassen:LICS06}.  The CPS approach consists in
proving a CPS-based correspondence between the bisimilarity $\rel_1$
we want to prove sound and a relation $\rel_2$ that we already know 
is a congruence. Because CPS translations are usually themselves
compatible, we can then conclude that $\rel_1$ is a congruence. For
example, for the $\lambda$-calculus, Lassen proved a
CPS-correspondence between the eager normal form bisimilarity and the
B{\"o}hm trees equivalence \cite{Lassen:LICS05}.

Because shift and reset have been originally defined in terms of CPS
\cite{Danvy-Filinski:LFP90}, one can expect the CPS approach to be
successful. However, the CPS translation of shift and reset assumes
that $\lamshift$ terms are executed within an outermost reset, and
therefore they cannot evaluate to a control stuck term. For the normal
form bisimilarity to be sound w.r.t. CPS, we would have to restrict
its definition to terms of the form $\reset \tm$. This does not seem
possible while keeping Definition \ref{d:nf-bisim} without
quantification over contexts. For example, to relate values $\valzero$
and $\valone$, we would have to execute $\appval \valzero \varx$ and
$\appval \valone \varx$ (where $\varx$ is fresh) under reset. However,
requiring simply $\reset{\appval \valzero \varx}$ and $\reset{\appval
  \valone \varx}$ to be related would be unsound; such a definition
would relate $\lam y {\shift \vark {\app \vark y}}$ and $\lam y
{\shift \vark {\app {\lamp z z} y}}$, which can be distinguished by
the context $\reset {\app \mtctx \app {\lamp z z} \Omega}$. To be
sound, we would have to require $\reset{\inctx \ctx {\appval \valzero
    \varx}}$ to be related to $\reset{\inctx \ctx {\appval \valone
    \varx}}$ \emph{for every $\ctx$}; we then introduce a
quantification over contexts that we want to avoid in the first
place. Because normal forms may contain control stuck terms as
sub-terms, normal form bisimilarity has to be able to handle them,
and, therefore, it cannot be restricted to terms of the form $\reset
\tm$ only.

Since CPS cannot help us in proving congruence, we follow a more
direct approach, by relying on a context closure. Given a relation
$\rel$, we define its substitutive, reflexive, and context closure
$\closc\rel$ by the rules of Fig. \ref{f:closc}. The main lemma of the
congruence proof is then as follows:

 \begin{figure}
  \begin{mathpar}
    \inferrule{ }
    {\tm \closc\rel \tm}
    \and
    \inferrule{\tmzero \rel \tmone}
    {\tmzero \closc\rel \tmone}
    \and
    \inferrule{\tmzero \closc\rel \tmone \\ \valzero \nf{\closc\rel} \valone}
    {\subst \tmzero \varx \valzero \closc\rel \subst \tmone \varx \valone}
    \and
    \inferrule{\tmzero \closc\rel \tmone \\ \rctxzero \closc\rel \rctxone}
    {\inctx \rctxzero \tmzero \closc\rel \inctx \rctxone \tmone}
    \and
    \inferrule{\tmzero \closc\rel \tmone}
    {\lam \varx \tmzero \closc\rel \lam \varx \tmone}
    \and
    \inferrule{\tmzero \closc\rel \tmone}
    {\shift \vark \tmzero \closc\rel \shift \vark \tmone}
  \end{mathpar}
  \caption{Substitutive, reflexive, and context closure of a relation $\rel$}
  \label{f:closc}
\end{figure}

\begin{lemma}
  \label{l:main-lemma-bisim}
  If $\rel$ is a normal form bisimulation, then so is $\closc\rel$.
\end{lemma}
More precisely, we prove that if $\tmzero \closc\rel \tmone$ and $\tmzero$
evaluates to some normal form $\tmzero'$ in $m$ steps, then $\tmone$ evaluates
to a normal form $\tmone'$ such that $\tmzero' \nf{\closc\rel} \tmone'$. The proof is
by nested induction on $m$ and on the definition of $\closc\rel$; it can be
found in Appendix~\ref{a:soundness}. Congruence of $\bisim$ then follows
immediately.
\begin{corollary}
  \label{c:congruence}
  The relation $\bisim$ is a congruence
\end{corollary}
We can then easily prove that $\bisim$ is sound w.r.t. contextual equivalence.
\begin{theorem}
  We have $\bisim \, \subseteq \, \ctxequiv$.
\end{theorem}
The following counter-example shows that the inclusion is in fact strict; normal
form bisimilarity is not complete.
\begin{proposition}
  \label{p:cex-dupl}
  Let $i=\lam y y$. We have $\reset{\app {\reset {\app \varx i}}{\shift \vark
      i}} \ctxequiv \reset{\app {\reset {\app \varx i}}{\appp{\reset {\app \varx
          i}}{\shift \vark i}}}$, but $\reset{\app {\reset {\app \varx
        i}}{\shift \vark i}} \not\bisim \reset{\app {\reset {\app \varx
        i}}{\appp{\reset {\app \varx i}}{\shift \vark i}}}$.
\end{proposition}

\begin{proof}
  Replacing $\varx$ by a closed value $\val$, we get $\reset{\app {\reset
      {\app \val i}}{\shift \vark i}}$ and $\reset{\app {\reset {\app \val
        i}}{\appp{\reset {\app \val i}}{\shift \vark i}}}$, which both evaluate
  to $i$ if the evaluation of $\reset {\app \val i}$ terminates (otherwise, they
  both diverge). With this observation, it is easy to prove that $\reset{\app
    {\reset {\app \varx i}}{\shift \vark i}}$ and $\reset{\app {\reset {\app
        \varx i}}{\appp{\reset {\app \varx i}}{\shift \vark i}}}$ are
  contextually equivalent. They are not bisimilar, because the terms
  $\reset{\app y {\shift \vark i}}$ and $\reset{\app y {\appp{\reset {\app \varx
          i}}{\shift \vark i}}}$ (where $y$ is fresh) are not bisimilar: the
  former evaluates to $i$ while the latter is in normal form (but is not a
  value).  \qed
\end{proof}

\section{Refined Bisimilarity and Up-to Techniques}
\label{s:refined-up-to}

In this section, we propose an improvement of the definition of normal
form bisimilarity, and we discuss some proof techniques which aim at
simplifying equivalence proofs.

\subsection{Refined Bisimilarity}
\label{ss:refined}

Normal form bisimilarity could better deal with control stuck terms. To
illustrate this, consider the following terms.
\begin{proposition}
  \label{p:cex-stuck}
  Let $i=\lam x x$. We have $\shift \vark i \ctxequiv \app{(\shift \vark i)}
  \Omega$, but $\shift \vark i \not\bisim \app{(\shift \vark i)} \Omega$.
\end{proposition}

\begin{proof}
  If $\shift \vark i$ and $\app{(\shift \vark i)} \Omega$ are put within a pure
  context, then we obtain two control stuck terms, and if we put these two terms
  within a context $\inctx \rctx {\reset \ctx}$, then they both
  reduce to $\inctx \rctx i$. Therefore, $\shift \vark i$ and $\app{(\shift
    \vark i)} \Omega$ are contextually equivalent. They are not normal form bisimilar, since
  the contexts $\mtctx$ and $\apctx \mtctx \Omega$ are not bisimilar
  ($\varx$ converges while $\app \varx \Omega$ diverges).  \qed
\end{proof}

When comparing control stuck terms, normal form bisimilarity considers
contexts and shift bodies separately, while they are combined if the
control stuck terms are put under a reset and the capture goes
through. To fix this issue, we consider another notion of
bisimulation. Given a relation $\rel$ on terms, we define $\rnf\rel$
on normal forms, which is defined the same way as $\nf\rel$ on values
and open stuck terms, and is defined on control stuck terms as
follows:
\begin{mathpar}
  \inferrule{\reset {\subst {\tmzero'} \vark {\lam \varx
        {\reset {\app {\vark'}{\inctx \ctxzero \varx}}}}} \rel \reset
    {\subst{\tmone'} \vark {\lam \varx {\reset{\app {\vark'}{\inctx \ctxone
              \varx}}}}} \\ \vark', \varx \textrm{
      fresh}}
  {\inctx \ctxzero {\shift \vark \tmzero} \rnf\rel \inctx \ctxone {\shift \vark
      \tmone}}
\end{mathpar}

\begin{definition}
  \label{d:refined-bisim}
  A relation $\rel$ on terms is a refined normal form simulation if $\tmzero \rel
  \tmone$ and $\tmzero \evalcbv \tmzero'$ implies $\tmone \evalcbv \tmone'$ and
  $\tmzero' \rnf\rel \tmone'$. A relation $\rel$ is a refined normal form bisimulation if
  both $\rel$ and $\inv\rel$ are refined normal form simulations. Refined normal form
  bisimilarity, written $\rbisim$, is the largest refined normal form bisimulation.
\end{definition}

In the control stuck terms case, Definition \ref{d:refined-bisim} simulates the capture
of $\ctxzero$ (respectively $\ctxone$) by $\shift \vark {\tmzero}$ (respectively
$\shift \vark {\tmone}$). However, if $\tmzero$ is put into a context
$\reset{\ctx}$, then $\shift \vark {\tmzero}$ captures a context bigger
than $\ctxzero$, namely $\inctx \ctx {\ctxzero}$. We take such
possibility into account by using a variable $\vark'$ in the definition of
$\rnf\rel$, which represents the context that can be captured beyond $\ctxzero$ and
$\ctxone$.

Refined bisimilarity contains the regular bisimilarity.
\begin{proposition}
  \label{p:bisim-subset-rbisim}
  We have $\bisim \, \subset \, \rbisim$.
\end{proposition}
Indeed, for control stuck terms, we have $\tmzero \evalcbv \inctx
\ctxzero {\shift \vark {\tmzero'}}$, $\tmone \evalcbv \inctx \ctxone
         {\shift \vark {\tmone'}}$, $\ctxzero \bisim \ctxone$, and
         $\reset{\tmzero'} \bisim \reset{\tmone'}$. Because $\bisim$
         is a congruence (Corollary \ref{c:congruence}), it is easy to
         see that $\reset {\subst {\tmzero'} \vark {\lam \varx {\reset
               {\app {\vark'}{\inctx \ctxzero \varx}}}}} \bisim \reset
         {\subst{\tmone'} \vark {\lam \varx {\reset{\app
                 {\vark'}{\inctx \ctxone \varx}}}}}$ holds for fresh
         $\vark'$ and $x$. Therefore, $\bisim$ is a refined
         bisimulation, and is included in $\rbisim$. The inclusion is
         strict, because $\rbisim$ relates the terms of Proposition
         \ref{p:cex-stuck}, while $\bisim$ does not.

Proving that $\rbisim$ is sound requires some adjustments to the
congruence proof of $\bisim$. First, given a relation $\rel$ on terms,
we define its substitutive, bisimilar, and context closure $\closbc
\rel$ by extending the rules of Fig. \ref{f:closc} with the following
one.
\begin{mathpar}
  \inferrule{\tmzero \rbisim \tmzero' \\ \tmzero' \closbc\rel \tmone' \\ \tmone' \rbisim \tmone}
  {\tmzero \closbc\rel \tmone}
\end{mathpar}
Henceforth, we simply write $\rbisim \closbc\rel \rbisim$ for the
composition of the three relations. Our goal is to prove that $\closbc
\rbisim$ is a refined bisimilarity. To this end, we need a few lemmas.
\begin{lemma}
  \label{l:beta-omega-ax}
  If $\varx \notin \fv \ctx$, then $\app {\lamp \varx {\inctx \ctx \varx}} \tm
  \bisim \inctx \ctx \tm$.
\end{lemma}
One can prove that $\{ (\app {\lamp \varx {\inctx \ctx \varx}} \tm,
\inctx \ctx \tm), \varx \notin \fv \ctx \} \cup \{(\tm,\tm)\}$ is a
bisimulation, by a straightforward case analysis on the result of the
evaluation of $\tm$ (if it exists). Note that Lemma
\ref{l:beta-omega-ax}, known as the $\beta_\Omega$ axiom in
\cite{Kameyama-Hasegawa:ICFP03}, has also been proved in
\cite{Biernacki-Lenglet:FOSSACS12} using applicative bisimulation. We
can see that the proof is much simpler using normal form
bisimulation. With Lemma \ref{l:beta-omega-ax}, congruence of
$\bisim$, and Proposition \ref{p:bisim-subset-rbisim}, we then prove
the following result.
\begin{lemma}
  \label{l:rbisim-trick}
  If $\varx \notin \fv \ctxzero \cup \fv \ctxone$ and $y \notin \fv
  \ctxone$ then $\reset{\subst \tm \vark {\lam \varx {\reset {\inctx
          \ctxone {\inctx \ctxzero \varx}}}}} \rbisim \reset{\subst
    \tm \vark {\lam \varx {\reset {\app {\lamp y {\inctx \ctxone
              y}}{\inctx \ctxzero \varx}}}}}$.
\end{lemma}
The main lemma of the congruence proof of $\rbisim$ is as follows.
\begin{lemma}
  \label{l:main-lemma-rbisim}
  If $\rel$ is a refined bisimulation, then so is $\closbc\rel$.
\end{lemma}
The proof is an adaptation of the proof of Lemma
\ref{l:main-lemma-bisim}. We sketch one sub-case of the proof, to
illustrate why we need $\closbc\rel$ (instead of $\closc\rel$) and
Lemma \ref{l:rbisim-trick}.
\begin{proof}[Sketch]
  Assume we are in the case where $\inctx \ctxzero \tmzero \closbc\rel
  \inctx \ctxone \tmone$ with $\inctx \ctxzero y \closbc\rel \inctx
  \ctxone y$ for a fresh $y$, and $\tmzero \closbc\rel
  \tmone$. Moreover, suppose $\tmzero \evalcbv
  \inctx{\ctxzero'}{\shift \vark {\tmzero'}}$. Then by the induction
  hypothesis, we know that there exist $\ctxone'$, $\tmone'$ such that
  $\tmone \evalcbv \inctx {\ctxone'}{\shift \vark {\tmone'}}$, and
  $\reset {\subst {\tmzero'} \vark {\lam \varx {\reset {\app
          {\vark'}{\inctx {\ctxzero'} \varx}}}}} \closbc\rel \reset
  {\subst{\tmone'} \vark {\lam \varx {\reset{\app {\vark'}{\inctx
            {\ctxone'} \varx}}}}}$ (*) for a fresh $\vark'$. Hence, we
  have $\inctx \ctxzero \tmzero \evalcbv \inctx
  \ctxzero{\inctx{\ctxzero'}{\shift \vark {\tmzero'}}}$ and $\tmone
  \evalcbv \inctx \ctxone {\inctx {\ctxone'}{\shift \vark
      {\tmone'}}}$, and we want to prove that $\reset {\subst
    {\tmzero'} \vark {\lam \varx {\reset {\app {\vark'}{\inctx
            \ctxzero {\inctx {\ctxzero'} \varx}}}}}} \closbc\rel
  \reset {\subst{\tmone'} \vark {\lam \varx {\reset{\app
          {\vark'}{\inctx \ctxone{\inctx {\ctxone'} \varx}}}}}}$
  holds. Because $\inctx \ctxzero y \closbc\rel \inctx \ctxone y$, we
  have $\lam y {\app {\vark'}{\inctx \ctxzero y}} \rnf{\closbc\rel}
  \lam y {\app {\vark'}{\inctx \ctxone y}}$ (**). Using (*) and (**),
  we obtain \[\reset {\subst {\tmzero'} \vark {\lam \varx {\reset {\app
          {\lamp y {\app {\vark'}{\inctx \ctxzero y}}}{\inctx
            {\ctxzero'} \varx}}}}} \closbc\rel \reset {\subst{\tmone'}
    \vark {\lam \varx {\reset{\app {\lamp y {\app {\vark'}{\inctx
                \ctxone y}}}{\inctx {\ctxone'} \varx}}}}},\] because
  $\closbc\rel$ is substitutive. By Lemma \ref{l:rbisim-trick}, we
  know that
  \begin{align*}
    \reset {\subst {\tmzero'} \vark {\lam \varx {\reset {\app {\vark'}{\inctx
              \ctxzero {\inctx {\ctxzero'} \varx}}}}}} & \rbisim \reset {\subst
      {\tmzero'} \vark {\lam \varx {\reset {\app {\lamp y {\app {\vark'}{\inctx
                  \ctxzero
                  y}}}{\inctx {\ctxzero'} \varx}}}}} \\
    \reset {\subst{\tmone'} \vark {\lam \varx {\reset{\app {\vark'}{\inctx
              \ctxone{\inctx {\ctxone'} \varx}}}}}} & \rbisim \reset
    {\subst{\tmone'} \vark {\lam \varx {\reset{\app {\lamp y {\app
                {\vark'}{\inctx \ctxone y}}}{\inctx {\ctxone'} \varx}}}}},
  \end{align*}
  which means that $\reset {\subst {\tmzero'} \vark {\lam \varx
      {\reset {\app {\vark'}{\inctx \ctxzero {\inctx {\ctxzero'}
              \varx}}}}}} \rbisim \closbc\rel \rbisim \reset
  {\subst{\tmone'} \vark {\lam \varx {\reset{\app {\vark'}{\inctx
            \ctxone{\inctx {\ctxone'} \varx}}}}}}$ holds. The required
  result then holds because $\rbisim \closbc\rel \rbisim \, \subseteq \, 
  \closbc\rel$.  \qed
\end{proof}
We can then conclude that $\rbisim$ is a congruence, and is sound
w.r.t. $\ctxequiv$.
\begin{corollary}
  The relation $\rbisim$ is a congruence.
\end{corollary}
\begin{theorem}
  We have $\rbisim \, \subset \, \ctxequiv$.
\end{theorem}
The inclusion is strict, because the terms of Proposition
\ref{p:cex-dupl} are still not related by $\rbisim$.

We would like to stress that even though $\rbisim$ equates more
contextually equivalent terms than $\bisim$,
the latter is still useful, since it leads to very simple proofs of
equivalence, as we can see with Lemma \ref{l:beta-omega-ax} (and with
the examples of Section~\ref{s:examples}). Therefore, $\rbisim$ does
not disqualify $\bisim$ as a proof technique.

\subsection{Up-to Techniques}

The idea behind up-to techniques
\cite{Sangiorgi-Walker:01,Lassen:98,Sangiorgi-al:LICS07} is to define
relations that are not exactly bisimulations but are included in
bisimulations. It usually leads to definitions of simpler candidate
relations and to simpler bisimulation proofs. As pointed out in
\cite{Lassen:98}, using a direct approach to prove congruence of the
normal form bisimilarity (as in Sections \ref{ss:soundness} and
\ref{ss:refined}) makes up-to techniques based on the context closure
easy to define and to prove valid. For example, we define bisimulation
up to substitutive, reflexive, and context closure (in short, up to
context) as follows.
\begin{definition}
  A relation $\rel$ on terms is a simulation up to context if $\tmzero
  \rel \tmone$ and $\tmzero \evalcbv \tmzero'$ implies $\tmone
  \evalcbv \tmone'$ and $\tmzero' \nf{\closc\rel} \tmone'$. A relation
  $\rel$ is a bisimulation up to context if both $\rel$ and $\inv\rel$
  are simulations up to context.
\end{definition}
Similarly, we can define a notion of refined bisimulation up to
context by replacing $\nf{\closc\rel}$ by $\rnf{\closbc\rel}$ in the
above definition. The proofs of Lemmas \ref{l:main-lemma-bisim} and
\ref{l:main-lemma-rbisim} can easily be adapted to bisimulations up to
context; a trivial change is needed only in the inductive case where
$\tmzero \closc\rel \tmone$ (respectively $\tmzero \closbc\rel
\tmone$) comes from $\tmzero \rel \tmone$.
\begin{lemma}
  If $\rel$ is a bisimulation up to context, then $\closc\rel$ is a
  bisimulation. If $\rel$ is a refined bisimulation up to context,
  then $\closbc\rel$ is a refined bisimulation.
\end{lemma}
Consequently, if $\rel$ is a bisimulation up to context, and if
$\tmzero \rel \tmone$, then $\tmzero \bisim \tmone$, because $\rel
\, \subseteq \, \closc\rel \, \subseteq \, \bisim$.

\begin{example}
  We can simplify the proof of bisimilarity between Turing's fixed
  point combinator and its delimited-control variant (cf. Example
  \ref{e:fixed-point}); indeed, it is enough to prove that $\rel = \{
  (\app \theta \theta, \reset {\app \theta {\shift \vark {\app \vark
        \vark}}}), (\app \theta \theta, \app{\lamp \varx {\reset{\app
        \theta \varx}}}{\lamp \varx {\reset{\app \theta \varx}}}) \}$
  is a bisimulation up to context.
\end{example}

When proving equivalence of terms, it is sometimes easier to reason in a
small-step fashion instead of trying to evaluate terms completely. To allow this
kind of reasoning, we define the following small-step notion.
\begin{definition}
  A relation $\rel$ on terms is a small-step simulation up to context
  if $\tmzero \rel \tmone$ implies:
  \begin{itemize}
  \item if $\tmzero \redcbv \tmzero'$, then there exists $\tmone'$ such that
    $\tmone \clocbv \tmone'$ and $\tmzero' \closc\rel \tmone'$;
  \item if $\tmzero$ is a normal form, then there exists $\tmone'$
    such that $\tmone \evalcbv \tmone'$ and $\tmzero \nf{\closc\rel}
    \tmone'$.
  \end{itemize}
  A relation $\rel$ is a small-step bisimulation up to context if both
  $\rel$ and $\inv\rel$ are small-step simulations up to context.
\end{definition}
Similarly, we can define the refined variant. Again, it is easy to
check the validity of these two proof techniques.
\begin{lemma}
  If $\rel$ is a small-step bisimulation up to context, then
  $\closc\rel$ is a bisimulation. If $\rel$ is a refined small-step
  bisimulation up to context, then $\closbc\rel$ is a refined
  bisimulation.
\end{lemma}
In the next section we show how these relations can be used
(Proposition~\ref{p:small-step}).

\section{Examples}
\label{s:examples}

We now illustrate the usefulness of the relations and techniques
defined in this paper, by proving some terms equivalences derived from
the axiomatization of $\lamshift$ \cite{Kameyama-Hasegawa:ICFP03}. The
relationship between contextual equivalence and Kameyama and
Hasegawa's axioms has been studied in
\cite{Biernacki-Lenglet:FOSSACS12}, using applicative bisimilarity. In
particular, we show that terms equated by all the axioms except for
$\AXshiftelim$ ($\shift \vark {\app \vark \tm}=\tm$ if $\vark \notin
\fv \tm$) are applicative bisimilar. The same result can be obtained
for normal form bisimilarity, using the same candidate relations as
for applicative bisimilarity (see Propositions 1 to 4 in
\cite{Biernacki-Lenglet:FOSSACS12}), except for the $\AXbetaomega$
axiom, where the equivalence proof becomes much simpler (see Lemma
\ref{l:beta-omega-ax}). The terms $\shift \vark {\app \vark \val}$ and
$\val$ (equated by $\AXshiftelim$) are not (applicative or normal
form) bisimilar, because the former is control stuck while the latter
is not. Conversely, there exist bisimilar terms that are not related
by the axiomatization, such as $\app \Omega \Omega$ and $\Omega$, or
Curry's and Turing's combinators (Example \ref{e:fixed-point}).

In this section, we propose several terms equivalences, the proofs of
which are quite simple using normal form bisimulation, especially
compared to applicative bisimulation. In the following, we write $\Id$
for the identity bisimulation $\{(\tm, \tm)\}$.
\begin{proposition}
  If $\varx \notin \fv \ctx$, then $\inctx \ctx {\app {\lamp \varx \tmzero}
    \tmone} \bisim \app {\lamp \varx {\inctx \ctx \tmzero}} \tmone$.
\end{proposition}
\begin{proof}
  By showing that $\{ (\inctx \ctx {\app {\lamp \varx \tmzero}
    \tmone}, \app {\lamp \varx {\inctx \ctx \tmzero}} \tmone), \varx
  \notin \fv\ctx \} \cup \Id$ is a normal form bisimulation. The proof
  is straightforward by case analysis on the result of the evaluation
  of $\tmone$ (if it exists).  \qed
\end{proof}
The next example demonstrates how useful small-step relations can be.
\begin{proposition}
  \label{p:small-step}
  If $\varx \notin \fv \ctx$, then $\reset{\app {\lamp \varx {\reset {\inctx
          \ctx \varx}}} \tm} \bisim \reset {\inctx \ctx \tm}$.
\end{proposition}
\begin{proof}
  Let $\rel = \{(\reset{\app {\lamp \varx {\reset {\inctx \ctx \varx}}} \tm},
  \reset {\inctx \ctx \tm}), \varx \notin \fv \ctx\}$. We prove that $\rel \cup
  \bisim$ is a small-step bisimulation up to context, by case analysis on $\tm$.
  \begin{itemize}
  \item If $\tm \redcbv \tm'$, then $\reset{\app {\lamp \varx {\reset {\inctx
            \ctx \varx}}} \tm} \redcbv \reset{\app {\lamp \varx {\reset {\inctx
            \ctx \varx}}}{\tm'}}$, $\reset {\inctx \ctx \tm} \redcbv \reset
    {\inctx \ctx {\tm'}}$, and we have $\reset{\app {\lamp \varx {\reset {\inctx
            \ctx \varx}}}{\tm'}} \rel \reset {\inctx \ctx {\tm'}}$, as required.
  \item If $\tm = \val$, then $\reset{\app {\lamp \varx {\reset {\inctx \ctx
            \varx}}} \val} \redcbv \reset{\reset {\inctx \ctx \val}}$. We have
    proved in Example \ref{e:reset-reset} that $\reset{\reset {\inctx \ctx
        \val}} \bisim \reset{\inctx \ctx \val}$.
  \item If $\tm = \inctx \rctx {\app y \val}$, then we have to relate
    $\reset{\app {\lamp \varx {\reset {\inctx \ctx \varx}}} \rctx }$
    and $\reset{\inctx \ctx \rctx }$ (we clearly have $\val \nf\bisim
    \val$). If $\rctx = \inctx {\rctx'}{\reset{\ctx'}}$, then we have
    $\reset{\app {\lamp \varx {\reset {\inctx \ctx \varx}}}{\inctx {\rctx'} z}}
    \rel \reset{\inctx \ctx {\inctx {\rctx'} z}}$ and $\reset{\inctx {\ctx'} z}
    \bisim \reset{\inctx{\ctx'} z}$ for a fresh $z$. If $\rctx = \ctx'$, then 
    $\reset{\app {\lamp \varx {\reset {\inctx \ctx \varx}}}{\inctx {\ctx'}
        z}} \rel \reset{\inctx \ctx {\inctx {\ctx'} z}}$ holds for a fresh $z$.
  \item If $\tm = \inctx {\ctx'}{\shift \vark {\tm'}}$, then $\reset{\app {\lamp
        \varx {\reset {\inctx \ctx \varx}}} \tm} \redcbv \reset{\subst{\tm'}
      \vark {\lam y {\reset {\app {\lamp \varx {\reset {\inctx \ctx
                  \varx}}}{\inctx {\ctx'} y}}}}}$, and $\reset {\inctx \ctx \tm}
    \redcbv \reset {\subst{\tm'} \vark {\lam y {\reset{\inctx \ctx
            {\inctx{\ctx'} y}}}}}$. We have $\reset {\app {\lamp \varx {\reset
          {\inctx \ctx \varx}}}{\inctx {\ctx'} y}} \rel \reset{\inctx \ctx
      {\inctx{\ctx'} y}}$, therefore $\reset{\subst{\tm'} \vark {\lam y
        {\reset {\app {\lamp \varx {\reset {\inctx \ctx \varx}}}{\inctx {\ctx'}
              y}}}}} \closc\rel \reset {\subst{\tm'} \vark {\lam y
        {\reset{\inctx \ctx {\inctx{\ctx'} y}}}}}$ holds, as wished. \qed
  \end{itemize}
\end{proof}
Without using small-step bisimulation, the definition of $\rel$ as
well as the bisimulation proof would be much more complex, since we
would have to compute the results of the evaluations of $\reset{\app
  {\lamp \varx {\reset {\inctx \ctx \varx}}} \tm}$ and of $\reset {\inctx
  \ctx \tm}$, which is particularly difficult if $\tm$ is a control
stuck term.

For the next example, we have to use refined bisimilarity.
\begin{proposition}
  \label{p:ex-refined}
  If $\vark' \notin \fv \ctx \cup \fv \tm$ and $\varx \notin \fv \ctx$, then we
  have 
  $\inctx \ctx {\shift \vark \tm} \rbisim \shift {\vark'}{\subst \tm \vark {\lam
      \varx {\reset {\app {\vark'}{\inctx \ctx \varx}}}}}$.
\end{proposition}

\begin{proof}
  The two terms are control stuck terms, therefore, we have to prove
  $\reset{\subst \tm \vark {\lam \varx {\reset {\app {\vark''}{\inctx \ctx
            \varx}}}}} \rbisim \reset{\subst \tm \vark {\lam \varx {\reset {\app
          {\lamp y {\reset {\app {\vark''} y}} }{\inctx \ctx \varx}}}}}$ for a
  fresh $\vark''$. We know that $\reset {\app {\vark''}{\inctx \ctx \varx}}
  \bisim \reset{\app {\lamp y {\reset {\app {\vark''} y}} }{\inctx \ctx \varx}}$
  holds by Proposition \ref{p:small-step}. Consequently, we have $\reset {\app
    {\vark''}{\inctx \ctx \varx}} \rbisim \reset{\app {\lamp y {\reset {\app
          {\vark''} y}} }{\inctx \ctx \varx}}$ by Proposition
  \ref{p:bisim-subset-rbisim}. We can then conclude by congruence of
  $\rbisim$. \qed
\end{proof}
Without Proposition \ref{p:small-step}, we would have to prove $\reset
{\app {\vark''}{\inctx \ctx \varx}} \rbisim \reset{\app {\lamp y
    {\reset {\app {\vark''} y}} }{\inctx \ctx \varx}}$ directly, using
a small-step refined bisimulation up to context. Proving Proposition
\ref{p:ex-refined} with the regular normal form bisimilarity would
require us to equate $\inctx \ctx y$ and $y$ (where $y$ is fresh),
which is not possible if $\ctx= \vctx {\lamp z \Omega} \mtctx$.

\section{Conclusion}
\label{s:conclusion}

In this paper, we propose several normal formal bisimilarities and
up-to techniques for a $\lambda$-calculus with shift and reset, and we
demonstrate their usefulness on a number of examples. Proving
equivalences of terms with the regular normal form bisimilarity
generates minimal proof obligations, especially when used in
conjunction with (small-step) up-to context techniques. If the regular
bisimilarity fails to relate the tested terms, then the refined
bisimilarity can be of help. If they both fail, then we may have to
use the applicative bisimilarity~\cite{Biernacki-Lenglet:FOSSACS12},
which, unlike the bisimilarities of this paper, is complete.

We believe this work can easily be adapted to other delimited-control
operators as well as the CPS hierarchy \cite{Danvy-Filinski:LFP90}. It
might also be interesting to extend this work to the typed
setting. Another possible future work would be to define
\emph{environmental bisimulations} \cite{Sangiorgi-al:LICS07} for
$\lamshift$. When comparing two terms, environmental relations use an
additional component, the environment, which represents the current
knowledge of the observer. For example, in the pure
$\lambda$-calculus, when two tested terms reduce to values, they
become known to the observer and are added to the environment. The
observer can then challenge two $\lambda$-abstractions by applying
them to two related arguments built from the
environment. Environmental bisimilarities are usually sound and
complete, and also allow for up-to techniques.

Another issue is to find a characterization of contextual equivalence
for $\lambda$-calculi with abortive control operators. Normal form
bisimilarities have been defined for extensions of the
$\lambda\mu$-calculus \cite{Lassen:LICS06}, but they are usually not
complete, except in the presence of a store construct
\cite{Stoevring-Lassen:POPL07}. It might be possible to reach
completeness with applicative or environmental bisimilarities.
\vspace{-3mm}
\paragraph*{Acknowledgments:}
We thank Ma{\l}gorzata Biernacka and the anonymous referees for
insightful comments on the presentation of this work.

\bibliographystyle{plain}
\bibliography{../../mybib}

\newpage

\appendix

\section{Soundness Proof}
\label{a:soundness}

\begin{lemma}
  \label{l:redcbv-subst}
  If $\tm \redcbv \tm'$ then $\subst \tm \varx \val \redcbv \subst {\tm'} \varx \val$.
\end{lemma}

\begin{proof}
  We proceed by case analysis on $\tm \redcbv \tm'$. 

  Suppose $\inctx \rctx {\app{\lamp y \tmzero}} \valzero \redcbv \inctx \rctx {\subst \tmzero y \valzero}$. We
  have 
  \begin{eqnarray*}
    \subst \tm \varx \val & =& \inctx {\subst \rctx \varx \val}{\app {\lamp y
        {\subst \tmzero \varx \val}}{\subst \valzero \varx \val}} \\
    & \redcbv& \inctx
    {\subst \rctx \varx \val}{\subst {\subst \tmzero \varx \val} y {\subst
        \valzero \varx \val}} = \subst {\tm'} \varx \val,
  \end{eqnarray*}
  as required.

  Suppose $\inctx \rctx {\reset{\inctx \ctxzero {\shift \vark \tmzero}}} \redcbv
  \inctx \rctx {\reset{\subst \tmzero \vark {\lam y {\reset {\inctx \ctxzero
            y}}}}}$. We have
  \begin{eqnarray*}
    \subst \tm \varx \val &=& \inctx{\subst \rctx \varx \val}{\reset{\inctx
        {\subst \ctxzero \varx \val}{\shift \vark {\subst \tmzero \varx \val}}}} \\
    &\redcbv& \inctx {\subst \rctx \varx \val}{\reset{\subst {\subst \tmzero \varx \val} \vark {\lam y {\reset
          {\inctx {\subst \ctxzero \varx \val} y}}}}} = \subst {\tm'} \varx \val,
  \end{eqnarray*}
  as required.

  Suppose $\inctx \rctx {\reset \valzero} \redcbv \inctx \rctx \valzero$. We
  have 
  $$\subst \tm \varx \val =
  \inctx{\subst \rctx \varx \val}{\reset {\subst \valzero \varx \val}} \redcbv
  \inctx{\subst \rctx \varx \val}{\subst \valzero \varx \val},$$ as
  required. 
  \qed
\end{proof}

\begin{lemma}
  \label{l:closc-properties}
  Let $\rel$ be a bisimulation. 
  \begin{itemize}
  \item If $\inctx \rctxzero {\app \varx \valzero} \nf{\closc\rel} \inctx
    \rctxone {\app \varx \valone}$ then $\inctx \rctxzero {\app \varx \valzero}
    \closc\rel \inctx \rctxone {\app \varx \valone}$ (and similarly for
    $\closbc\rel$).
  \item If $\lam \varx \tmzero \nf{\closc\rel} \lam \varx \tmone$ then $\tmzero
    \closc\rel \tmone$ (and similarly for $\closbc\rel$).
  \end{itemize}
\end{lemma}

\begin{proof}
  The relation $\inctx \rctxzero {\app \varx \valzero} \nf{\closc\rel} \inctx
  \rctxone {\app \varx \valone}$ implies $\rctxzero \closc\rel \rctxone$ and
  $\valzero \nf{\closc\rel} \valone$. We have $\app \varx y \closc\rel \app
  \varx y$ for a fresh $y$, therefore we have $\inctx \rctxzero{\app \varx y}
  \closc\rel \inctx \rctxone {\app \varx y}$, which in turn implies $\inctx
  \rctxzero{\app \varx \valzero} \closc\rel \inctx \rctxone {\app \varx
    \valone}$.

  The second item is easy by definition of $\lam \varx \tmzero \nf{\closc\rel}
  \lam \varx \tmone$. \qed
\end{proof}

\begin{lemma}[Lemma \ref{l:main-lemma-bisim} in the paper]
  \label{l:closc-bisim}
  If $\rel$ is a bisimulation, then $\closc\rel$ is a bisimulation.
\end{lemma}

\begin{proof}
  Because $\closc\rel$ is symmetric, we only have to prove that it is
  a simulation. We consider $\tmzero \closc\rel \tmone$ with $\tmzero$
  evaluating in $m$ steps; we prove that $\tmone$ evaluates to a term
  related by $\nf{\closc\rel}$ by induction on $m$, and on the
  derivation of $\tmzero \closc\rel \tmone$, ordered
  lexicographically. The case $m=0$ is easy by induction on $\tmzero
  \closc\rel \tmone$; we treat only the general case $m > 0$. Note
  that the cases $\lam \varx \tmzero \closc\rel \lam \varx \tmone$
  with $\tmzero \closc\rel \tmone$ and $\shift \vark \tmzero
  \closc\rel \shift \vark \tmone$ with $\tmzero \closc\rel \tmone$ are
  not treated here since they are part of the base case.
  
  \paragraph{Assume we have $\tmzero \rel \tmone$.} This case is easy because
  $\rel$ is a bisimulation and $\rel \subseteq \closc\rel$. The case $\tmzero
  \closc\rel \tmzero$ is also easy.
  
  \paragraph{Assume $\subst \tmzero \varx \valzero \closc\rel \subst \tmone
    \varx \valone$ with $\tmzero \closc\rel \tmone$ and $\valzero \nf{\closc\rel}
    \valone$.} We suppose first that $\subst \tmzero \varx \valzero \evalcbv
  \valzero'$. We have two cases to consider.
  \begin{itemize}
  \item If $\tmzero \evalcbv \valzero''$, then $\valzero' =
    \subst{\valzero''} \varx \valzero$ by Lemma
    \ref{l:redcbv-subst}. By the induction hypothesis, there exists
    $\valone''$ such that $\tmone \evalcbv \valone''$, and $\valzero''
    \nf{\closc\rel} \valone''$ holds. By Lemma \ref{l:redcbv-subst},
    we have $\subst \tmone \varx \valone \evalcbv \subst{\valone''}
    \varx \valone$, and we also have $\subst {\valzero''} \varx
    \valzero \nf{\closc\rel} \subst {\valone''} \varx \valone$, hence
    the result holds.
  \item Suppose $\tmzero \evalcbv \inctx \rctxzero {\app \varx
    {\valzero''}}$ with $\inctx {\subst \rctxzero \varx \valzero}{\app
    \valzero {\subst {\valzero''} \varx \valzero}} \evalcbv
    \valzero'$. By the induction hypothesis, there exist $\rctxone$,
    $\valone''$ such that $\tmone \evalcbv \inctx \rctxone {\app \varx
      {\valone''}}$, $\rctxzero \closc\rel \rctxone$, and $\valzero''
    \nf{\closc\rel} \valone''$. Because $\inctx {\subst \rctxzero
      \varx \valzero}{\app \valzero {\subst {\valzero''} \varx
        \valzero}}$ evaluates to a value, $\valzero$ must be a
    $\lambda$-abstraction $\valzero = \lam y {\tmzero'}$. Assume
    $\valone$ is a variable $y$. Since $\valzero \nf{\closc\rel}
    \valone$, we have $\appval \valzero z \closc\rel \app y z$ for a
    fresh $z$; because $\app y z$ is in normal form, we can apply the
    induction hypothesis with $m=0$. There exist $\ctxzero'$,
    $\valzero'''$ such that $\appval \valzero z \evalcbv \inctx
    {\ctxzero'}{\app y {\valzero'''}}$, $\inctx {\ctxzero'}{\varx'}
    \closc\rel \varx'$ for a fresh $\varx'$, and $\valzero'''
    \nf{\closc\rel} z$. Consequently, we have
    $$\inctx {\subst \rctxzero \varx \valzero}{\app
      \valzero {\subst {\valzero''} \varx \valzero}} \evalcbv
    \inctx {\subst \rctxzero \varx \valzero}{\inctx
      {\subst{\ctxzero'} z {\subst {\valzero''} \varx
          \valzero}}{\app y {\subst {\valzero'''} z {\subst
            {\valzero''} \varx \valzero}}}},$$ which is in
             contradiction with $\inctx {\subst \rctxzero \varx
               \valzero}{\app \valzero {\subst {\valzero''} \varx
                 \valzero}} \evalcbv \valzero'$. Therefore,
             $\valone$ must be a $\lambda$-abstraction $\lam y
             {\tmone'}$.

    By Lemma \ref{l:closc-properties}, we have $\tmzero' \closc\rel \tmone'$. The
    reductions 
    $$\inctx {\subst \rctxzero \varx \valzero}{\app \valzero {\subst
        {\valzero''} \varx \valzero}} \redcbv \inctx {\subst \rctxzero \varx
      \valzero}{\subst {\tmzero'} y {\subst {\valzero''} \varx \valzero}}$$
    and
    $$\inctx {\subst \rctxone \varx \valone}{\app \valone {\subst
        {\valone''} \varx \valone}} \redcbv \inctx {\subst \rctxone \varx
      \valone}{\subst {\tmone'} y {\subst {\valone''} \varx \valone}}$$ hold.
    Because $\inctx {\subst \rctxzero \varx \valzero}{\subst {\tmzero'} y
      {\subst {\valzero''} \varx \valzero}} \closc\rel \inctx {\subst \rctxone
      \varx \valone}{\subst {\tmone'} y {\subst {\valone''} \varx \valone}}$,
    and $\inctx {\subst \rctxzero \varx \valzero}{\subst {\tmzero'} y {\subst
        {\valzero''} \varx \valzero}}$ evaluates to $\valzero'$ in less than
    $m-1$ steps, we can apply the induction hypothesis. Therefore, there exists $\valone'$
    such that $\inctx {\subst \rctxone \varx \valone}{\subst {\tmone'} y {\subst
        {\valone''} \varx \valone}} \evalcbv \valone'$ and $ \valzero'
    \nf{\closc\rel} \valone'$. One can check that we have $\subst \tmone \varx \valone \evalcbv
    \valone'$, hence the result holds.
  \end{itemize}

  The case $\subst \tmzero \varx \valzero \evalcbv \inctx \ctxzero {\shift \vark
    {\tmzero'}}$ is treated similarly. Suppose $\subst \tmzero \varx \valzero
  \evalcbv \inctx \rctxzero {\app y {\valzero'}}$ with $y \neq x$. We have two
  possible cases. The case $\tmzero \evalcbv \inctx {\rctxzero'}{\app y
    {\valzero''}}$ is similar to the case $\subst \tmzero \varx \valzero
  \evalcbv \valzero'$ with $\tmzero \evalcbv \valzero$. Suppose $\tmzero
  \evalcbv \inctx {\rctxzero'}{\app \varx {\valzero''}}$ with $\inctx {\subst
    {\rctxzero'} \varx \valzero}{\app \valzero {\subst {\valzero''} \varx
      \valzero}} \evalcbv \inctx \rctxzero {\app y {\valzero'}}$. By the induction hypothesis,
  there exist $\rctxone'$, $\valone''$ such that $\tmone \evalcbv \inctx
  {\rctxone'}{\app \varx {\valone''}}$, $\rctxzero' \closc\rel \rctxone'$, and
  $\valzero'' \nf{\closc\rel} \valone''$. If both $\valzero$ and $\valone$ are
  $\lambda$-abstractions, then we proceed as in the case $\subst \tmzero \varx
  \valzero \evalcbv \valzero'$ with $\tmzero \evalcbv \inctx \rctxzero {\app
    \varx {\valzero''}}$. If both $\valzero$ and $\valone$ are variables, then
  we must have $\valzero=\valone=y$, and the required result holds. Suppose
  $\valzero$ is a variable and $\valone$ is a $\lambda$-abstraction (the
  symmetric case is treated similarly). Then we must have $\valzero = y$,
  $\rctxzero = \subst {\rctxzero'} \varx \valzero$, and $\valzero' = \subst
  {\valzero''} \varx \valzero$. Because $\valzero \nf{\closc\rel} \valone$, we
  have $\app y z \closc\rel \appval \valone z$ for a fresh $z$, so by the induction hypothesis
  (case $m=0$) there exist $\ctxone'$, $\valone'''$ such that $\appval \valone z
  \evalcbv \inctx {\ctxone'}{\app y {\valone'''}}$, $\inctx {\ctxone'}{\varx'}
  \closc\rel \varx'$ for a fresh $\varx'$, and $\valone''' \nf{\closc\rel}
  z$. Consequently, we have $\subst \tmone \varx \valone \evalcbv \inctx {\subst
    {\rctxone'} \varx \valzero}{\inctx {\subst{\ctxone'} z {\subst {\valone''}
        \varx \valone}}{\app y {\subst {\valone'''} z {\subst {\valone''} \varx
          \valone}}}}$. From the relations $\rctxzero' \closc\rel \rctxone'$ and
  $\varx' \closc\rel \inctx {\ctxone'}{\varx'}$, we deduce $\subst {\rctxzero'}
  \varx \valzero \closc\rel \inctx {\subst {\rctxone'} \varx
    \valzero}{\subst{\ctxone'} z {\subst {\valone''} \varx \valone}}$. From
  $\valzero'' \nf{\closc\rel} \valone''$ and $z \nf{\closc\rel} \valone'''$, we
  deduce $\subst {\valzero''} \varx \valzero \nf{\closc\rel} \subst {\valone'''} z
  {\subst {\valone''} \varx \valone}$. Consequently we have the required result.

  \paragraph{Assume $\inctx \ctxzero \tmzero \closc\rel \inctx \ctxone \tmone$
    with $\tmzero \closc\rel \tmone$ and $\inctx \ctxzero \varx
    \closc\rel \inctx \ctxzero \varx$ for a fresh $\varx$.} Suppose
  $\inctx \ctxzero \tmzero \evalcbv \valzero$. Then $\tmzero \evalcbv
  \valzero'$ and $\inctx \ctxzero {\valzero'} \evalcbv \valzero$. By
  the induction hypothesis, there exists $\valone'$ such that $\tmone
  \evalcbv \valone'$, and $\valzero' \nf{\closc\rel}
  \valone'$. Because $\inctx \ctxzero {\valzero'} \evalcbv$, there
  exists a normal form $\tmzero'$ such that $\inctx \ctxzero \varx
  \evalcbv \tmzero'$. By the induction hypothesis, there exists a normal
  form $\tmone'$ such that $\inctx \ctxone \varx \evalcbv \tmone'$ and
  $\tmzero' \nf{\closc\rel} \tmone'$. By Lemma \ref{l:redcbv-subst},
  we have $\inctx \ctxzero \tmzero \clocbv \subst {\tmzero'} \varx
  {\valzero'}$ and $\inctx \ctxone \tmone \clocbv \subst {\tmone'}
  \varx {\valone'}$. Suppose $\inctx \ctxzero \tmzero$ reduces to
  $\subst {\tmzero'} \varx {\valzero'}$ in at least one step. Then
  $\subst {\tmzero'} \varx {\valzero'}$ evaluates to $\valzero$ in
  strictly less than $m$ steps. The normal form $\tmzero'$ is either a
  value or an open stuck term. If $\tmzero'$ is a value, then $\subst
  {\tmone'} \varx {\valone'}$ is also a value. From $\tmzero'
  \nf{\closc\rel} \tmone'$ and substitutivity of $\closc\rel$, we can
  prove that $\subst {\tmzero'} \varx {\valzero'} \nf{\closc\rel}
  \subst {\tmone'} \varx {\valone'}$ holds, as wished. If $\tmzero'$
  is an open stuck term, then so is $\tmone'$, and $\subst {\tmzero'}
  \varx {\valzero'} \closc\rel \subst {\tmone'} \varx {\valone'}$
  holds by Lemma~\ref{l:closc-properties} and substitutivity of
  $\closc\rel$. By the induction hypothesis, there exists $\valone$ such
  that $\subst {\tmone'} \varx {\valone'} \evalcbv \valone$, and
  $\valzero \nf{\closc\rel} \valone$. One can check that $\inctx
  \ctxone \tmone \evalcbv \valone$, hence we have the required
  result. Suppose now that $\inctx \ctxzero \tmzero = \subst
  {\tmzero'} \varx {\valzero'}$. It is possible only if $\tmzero =
  \valzero'$ and $\inctx \ctxzero \varx = \inctx {\ctxzero'}{\app
    \varx {\valzero''}}$. Then we have $\tmone' =
  \inctx{\ctxone'}{\app \varx {\valone''}}$ with $\inctx
        {\ctxzero'}{\varx'} \closc\rel \inctx {\ctxone'}{\varx'}$ for
        a fresh $\varx'$ and $\valzero'' \nf{\closc\rel}
        \valone''$. Because $\inctx {\ctxzero'}{\app
          {\valzero'}{\valzero''}}$ evaluates to $\valzero$,
        $\valzero'$ must be $\lambda$-abstraction $\lam z
                 {\tmzero''}$. By a similar reasoning as in the case
                 $\subst \tmzero \varx \valzero \evalcbv \valzero'$
                 (second sub-case), $\valone'$ is also a
                 $\lambda$-abstraction $\lam z {\tmone''}$. By
                 Lemma~\ref{l:closc-properties}, we have $\tmzero''
                 \closc\rel \tmone''$. Therefore we have $\inctx
                 \ctxzero \tmzero \redcbv \inctx {\ctxzero'}{\subst
                   {\tmzero''} z {\valzero''}}$ and $\inctx \ctxone
                 \tmone \clocbv \inctx {\ctxone'}{\subst {\tmone''} z
                   {\valone''}}$. From $\valzero'' \nf{\closc\rel}
                 \valone''$, we obtain $\inctx {\ctxzero'}{\subst
                   {\tmzero''} z {\valzero''}} \closc\rel \inctx
                          {\ctxone'}{\subst {\tmone''} z {\valone''}}$
                          by substitutivity. Because $\inctx
                          {\ctxzero'}{\subst {\tmzero''} z {\valzero''}}$
                          evaluates to $\valzero$ in less than $m-1$
                          steps, by the induction hypothesis (on $m$),
                          there exists $\valone$ such that $\inctx
                          {\ctxone'}{\subst {\tmone''} z {\valone''}}
                          \evalcbv \valone$ and $\valzero
                          \nf{\closc\rel} \valone$. One can check that
                          $\inctx \ctxone \tmone \evalcbv \valone$
                          holds, hence we have the required result.

  The case $\inctx \ctxzero \tmzero \evalcbv \inctx {\ctxzero'}{\shift \vark
    {\tmzero'}}$ is similar. Suppose now that $\inctx \ctxzero \tmzero \evalcbv
  \inctx \rctxzero {\app y \valzero}$. We have two possible cases; the case
  $\tmzero \evalcbv \inctx {\rctxzero'}{\app y \valzero}$ is easy using
  induction. Suppose $\tmzero \evalcbv \valzero'$ and $\inctx
  \ctxzero {\valzero'} \evalcbv \inctx {\rctxzero}{\app y \valzero}$. By
  the induction hypothesis, there exists $\valone'$ such that $\tmone \evalcbv
  \valone'$, and $\valzero' \nf{\closc\rel} \valone'$. Because $\inctx \ctxzero
  {\valzero'} \evalcbv$, there exists a normal form $\tmzero'$ such that $\inctx
  \ctxzero \varx \evalcbv \tmzero'$. By the induction hypothesis, there exists a
  normal form $\tmone'$ such that $\inctx \ctxone \varx \evalcbv \tmone'$ and
  $\tmzero' \nf{\closc\rel} \tmone'$. By Lemma \ref{l:redcbv-subst}, we have
  $\inctx \ctxzero \tmzero \clocbv \subst {\tmzero'} \varx {\valzero'}$ and
  $\inctx \ctxone \tmone \clocbv \subst {\tmone'} \varx {\valone'}$. If $\inctx
  \ctxzero \tmzero$ reduces to $\subst {\tmzero'} \varx {\valzero'}$ in at least 
  one step, then we proceed as in the case $\inctx \ctxzero \tmzero
  \evalcbv \valzero$. Otherwise, we have $\tmzero = \valzero'$, $\inctx \ctxzero
  \varx = \inctx {\ctxzero'}{\app \varx {\valzero''}}$, $\tmone' =
  \inctx{\ctxone'}{\app \varx {\valone''}}$ with $\inctx {\ctxzero'}{\varx'}
  \closc\rel \inctx {\ctxone'}{\varx'}$ for a fresh $\varx'$, and $\valzero''
  \nf{\closc\rel} \valone''$. If both $\valzero'$ and $\valone'$ are
  $\lambda$-abstractions, then we proceed as in the case $\inctx \ctxzero
  \tmzero \evalcbv \valzero$. If they are both variables, then
  $\valzero'=\valone'=y$, and the result holds. If $\valzero'$ is a variable and
  $\valone'$ is a $\lambda$-abstraction, then we must have $\valzero' = y$,
  $\rctxzero = \ctxzero''$, and $\valzero = \valzero''$. Because we have $\app y
  z \closc\rel \appval \valone z$, by the induction hypothesis (case $m=0$), there exist
  $\ctxone''$, $\valone'''$ such that $\appval \valone z \evalcbv \inctx
  {\ctxone''}{\app y {\valone'''}}$, $y' \closc\rel \inctx {\ctxone''}{y'}$ for
  a fresh $y'$, and $z \nf{\closc\rel} \valone'''$. Consequently, we have $\inctx
  \ctxone \tmone \evalcbv \inctx {\ctxone'}{\inctx {\subst {\ctxone''} z
      {\valone''}}{\app y {\subst {\valone'''} z {\valone''}}}}$. From $\inctx
  {\ctxzero'}{\varx'} \closc\rel \inctx {\ctxone'}{\varx'}$ and $y' \closc\rel
  \inctx {\ctxone''}{y'}$, we deduce $\ctxzero' \closc\rel \inctx
  {\ctxone'}{\subst {\ctxone''} z {\valone''}}$. From $\valzero'' \nf{\closc\rel}
  \valone''$ and $z \nf{\closc\rel} \valone'''$, we deduce $\valzero'' \nf{\closc\rel}
  \subst {\valone'''} z {\valone''}$. Consequently, we have the required result.

  \paragraph{Assume $\inctx \rctxzero {\reset {\inctx \ctxzero \tmzero}} \closc\rel
    \inctx \rctxone {\reset {\inctx \ctxone \tmone}}$ with $\inctx \rctxzero
    \varx \closc\rel \inctx \rctxone \varx$, $\inctx \ctxzero \varx \closc\rel \inctx
    \ctxone \varx$ ($\varx$ fresh), and $\tmzero \closc\rel \tmone$.} Note that
  $\inctx \rctxzero {\reset {\inctx \ctxzero \tmzero}}$ cannot evaluate to
  $\inctx{\ctxzero'}{\shift \vark {\tmzero'}}$. Suppose $\inctx
  \rctxzero {\reset {\inctx \ctxzero \tmzero}} \evalcbv \valzero$. We have
  several cases to consider.
  \begin{itemize}
  \item Suppose $\tmzero \evalcbv \valzero'$, $\inctx \ctxzero {\valzero'}
    \evalcbv \valzero''$, and $\inctx \rctxzero {\reset {\valzero''}} \evalcbv
    \valzero$. By the induction hypothesis, there exists $\valone'$ such that $\tmone \evalcbv
    \valone'$ and $\valzero' \nf{\closc\rel} \valone'$. We have $\inctx \ctxone
    \tmone \clocbv \inctx \ctxone {\valone'}$ and $\inctx \ctxzero {\valzero'}
    \closc\rel \inctx \ctxone {\valone'}$. Because the evaluation $\inctx
    \rctxzero {\reset {\valzero''}} \evalcbv \valzero$ takes at least one step
    (corresponding to $\reset {\valzero''} \redcbv \valzero''$), we know that
    the evaluation $\inctx \ctxzero {\valzero'} \evalcbv \valzero''$ is in $m-1$
    steps or less. Therefore, by the induction hypothesis (on $m$), there exists $\valone''$
    such that $\inctx \ctxone {\valone'} \evalcbv \valone''$ and $\valzero''
    \nf{\closc\rel} \valone''$. Because $\inctx \rctxzero {\reset {\valzero''}}
    \evalcbv \valzero$, there exists a normal form $\tmzero'$ such that $\inctx
    \rctxzero \varx \evalcbv \tmzero'$. By the induction hypothesis, there exists a normal form
    $\tmone'$ such that $\inctx \rctxone \varx \evalcbv \tmone'$ and $\tmzero'
    \nf{\closc\rel} \tmone'$. We have $\inctx \rctxzero {\reset {\valzero''}}
    \clocbv \subst {\tmzero'} \varx {\valzero''}$ and $\inctx \rctxone {\reset
      {\valone''}} \clocbv \subst {\tmone'} \varx {\valone''}$. Because the
    reduction $\inctx \rctxzero {\reset {\valzero''}} \clocbv \subst {\tmzero'}
    \varx {\valzero''}$ takes at least one step, we know that the evaluation
    $\subst {\tmzero'} \varx {\valzero''} \evalcbv \valzero$ takes $m-1$ steps
    or less. Besides, $\tmzero'$ is either a value or an open stuck term. If
    $\tmzero'$ is a value, then so is $\tmone'$, and one can check that both
    $\valzero = \subst {\tmzero'} \varx {\valzero''} \nf{\closc\rel} \subst
    {\tmone'} \varx {\valone''}$ and $\inctx \rctxone {\reset {\inctx \ctxone
        \tmone}} \evalcbv \subst {\tmone'} \varx {\valone''}$ hold. If
    $\tmzero'$ is an open term, then so is $\tmone'$, and we have
    $\subst{\tmzero'} \varx {\valzero''} \closc\rel \subst{\tmone'} \varx
    {\valone''}$ by Lemma \ref{l:closc-properties} and substitutivity. Therefore,
    by the induction hypothesis (on $m$), there exists $\valone$ such that $\subst {\tmone'}
    \varx {\valone''} \evalcbv \valone$ and $\valzero \nf{\closc\rel}
    \valone$. Because $\inctx \rctxone {\reset {\inctx \ctxone \tmone}} \evalcbv
    \valone$, we have the required result.
  \item Suppose $\tmzero \evalcbv \valzero'$, $\inctx \ctxzero {\valzero'}
    \evalcbv \inctx {\ctxzero'}{\shift \vark {\tmzero'}}$, and $\inctx \rctxzero
    {\reset{\inctx {\ctxzero'}{\shift \vark {\tmzero'}}}} \evalcbv \valzero$. By
    the induction hypothesis, there exists $\valone'$ such that $\tmone \evalcbv \valone'$ and
    $\valzero' \nf{\closc\rel} \valone'$. We have $\inctx \ctxone \tmone \clocbv
    \inctx \ctxone {\valone'}$ and $\inctx \ctxzero {\valzero'} \closc\rel
    \inctx \ctxone {\valone'}$. Because the evaluation $\inctx \rctxzero
    {\reset{\inctx {\ctxzero'}{\shift \vark {\tmzero'}}}} \evalcbv \valzero$
    takes at least one step (corresponding to the capture of $\ctxzero'$ by
    shift), we know that the evaluation $\inctx \ctxzero {\valzero'} \evalcbv
    \inctx {\ctxzero'}{\shift \vark {\tmzero'}}$ is in $m-1$ steps or
    less. Therefore, by the induction hypothesis (on $m$), there exists $\ctxone'$, $\tmone'$
    such that $\inctx \ctxone {\valone'} \evalcbv \inctx {\ctxone'}{\shift \vark
      {\tmone'}}$, $\reset {\tmzero'} \closc\rel \reset{\tmone'}$, and $\inctx
    {\ctxzero'} y \closc\rel \inctx {\ctxone'} y$ for a fresh $y$. By
    congruence, we have $\lam y {\reset{\inctx {\ctxzero'} y}} \nf{\closc\rel} \lam y
    {\reset{\inctx {\ctxone'} y}}$, therefore, $\reset {\subst {\tmzero'} \vark
      {\lam y {\reset {\inctx {\ctxzero'} y}}}} \closc\rel \reset{\subst
      {\tmone'} \vark {\lam y {\reset {\inctx {\ctxone'} y}}}}$ holds by
    substitutivity. Because $\inctx \rctxzero {\reset{\inctx {\ctxzero'}{\shift
          \vark {\tmzero'}}}}$ evaluates to $\valzero$, we must have
    $\reset{\inctx {\ctxzero'}{\shift \vark {\tmzero'}}} \evalcbv \valzero''$,
    $\inctx \rctxzero \varx \evalcbv \tmzero''$ for some normal form
    $\tmzero''$, and $\subst {\tmzero''} \varx {\valzero''} \evalcbv
    \valzero$. Because of the capture step $\reset{\inctx {\ctxzero'}{\shift
        \vark {\tmzero'}}} \redcbv \reset {\subst {\tmzero'} \vark {\lam y
        {\reset {\inctx {\ctxzero'} y}}}}$, we know $\reset {\subst {\tmzero'}
      \vark {\lam y {\reset {\inctx {\ctxzero'} y}}}}$ evaluates to $\valzero''$
    in $m-1$ steps or less. Consequently, by the induction hypothesis (on $m$), there exists
    $\valone''$ such that $\reset{\subst {\tmone'} \vark {\lam y {\reset {\inctx
            {\ctxone'} y}}}} \evalcbv \valone''$ and $\valzero'' \nf{\closc\rel}
    \valone''$. Because $\inctx \rctxzero \varx \closc\rel \inctx \rctxone
    \varx$, we also know by the induction hypothesis that there exists a normal form $\tmone''$
    such that $\inctx \rctxone \varx \evalcbv \tmone''$ and $\tmzero''
    \nf{\closc\rel} \tmone''$. Because the reduction $\inctx \rctxzero
    {\reset{\inctx {\ctxzero'}{\shift \vark {\tmzero'}}}} \clocbv \subst
    {\tmzero''} \varx {\valzero''}$ takes at least one step, we know that the
    evaluation $\subst {\tmzero''} \varx {\valzero''} \evalcbv \valzero$ is in
    $m-1$ steps or less. Besides, $\tmzero''$ is either a value or an open stuck
    term. If $\tmzero''$ is a value, then so is $\tmone''$, and one can check
    that both $\valzero = \subst{\tmzero''} \varx {\valzero''} \nf{\closc\rel}
    \subst{\tmone''} \varx {\valone''}$ and $\inctx \rctxone {\reset {\inctx
        \ctxone \tmone}} \evalcbv \subst{\tmone''} \varx {\valone''}$ hold. If
    $\tmzero''$ is an open stuck term, then so is $\tmone''$, and we have
    $\subst {\tmzero''} \varx {\valzero''} \closc\rel \subst{\tmone''} \varx
    {\valone''}$ by Lemma \ref{l:closc-properties} and substitutivity. By
    induction (on $m$), there exists $\valone$ such that $\subst {\tmone''}
    \varx {\valone''} \evalcbv \valone$ and $\valzero \nf{\closc\rel}
    \valone$. One can check that $\inctx \rctxone {\reset {\inctx \ctxone
        \tmone}} \evalcbv \valone$ holds, therefore the required result holds.
  \item Suppose $\tmzero \evalcbv \inctx {\ctxzero'}{\shift \vark \tmzero}$ and
    $\inctx \rctxzero {\reset{\inctx \ctxzero {\inctx {\ctxzero'}{\shift \vark
            {\tmzero'}}}}} \evalcbv \valzero$. This sub-case is similar to the
    previous one.
  \end{itemize}

  Suppose $\inctx \rctxzero {\reset {\inctx \ctxzero \tmzero}} \evalcbv \inctx
  {\rctxzero'}{\app y \valzero}$. There are five sub-cases to consider: three of
  them are similar to the sub-cases of $\inctx \rctxzero {\reset {\inctx
      \ctxzero \tmzero}} \evalcbv \valzero$, and the remaining two are similar
  to the sub-cases of $\inctx \ctxzero \tmzero \evalcbv \inctx \rctxzero {\app y
    \valzero}$ (namely $\tmzero \evalcbv \inctx {\rctxzero''}{\app y \valzero}$
  with $\rctxzero' = \inctx \rctxzero {\reset {\inctx \ctxzero {\rctxzero''}}}$,
  or $\tmzero \evalcbv \valzero'$, $\inctx \ctxzero {\valzero'} \evalcbv \inctx
  {\rctxzero''}{\app y \valzero}$ with $\rctxzero' = \inctx \rctxzero {\reset
    {\rctxzero''}}$).

  \qed
\end{proof}

\begin{lemma}[Lemma \ref{l:main-lemma-rbisim} in the paper]
  If $\rel$ is a refined bisimulation, then so is $\closbc\rel$.
\end{lemma}
Because the proof is quite similar to the previous one, we sketch only the cases
with the largest differences.

\begin{proof}[Sketch]
  Assume $\tmzero \rbisim \tmzero^1 \closbc\rel \tmone^1 \rbisim \tmone$ and
  $\tmzero \evalcbv \tmzero'$, where $\tmzero'$ is a normal form. By
  bisimilarity, there exists $\tmzero''$ such that $\tmzero^1 \evalcbv
  \tmzero''$ and $\tmzero' \rnf\rbisim \tmzero''$. By the induction hypothesis (on the definition
  of $\closbc\rel$), there exists $\tmone''$ such that $\tmone^1 \evalcbv
  \tmone''$ and $\tmzero'' \rnf{\closbc\rel} \tmone''$. By bisimilarity, there exists
  $\tmone'$ such that $\tmone \evalcbv \tmone'$ and $\tmone'' \rnf{\rbisim}
  \tmone'$. Finally, we have $\tmzero' \rnf{(\rbisim \closbc\rel \rbisim)} \tmone'$, and
  because $\rbisim \closbc\rel \rbisim \subseteq \closbc\rel$, we have the
  required result.

  Assume we are in the case where $\inctx \ctxzero \tmzero \closbc\rel \inctx
  \ctxone \tmone$ with $\inctx \ctxzero y \closbc\rel \inctx \ctxone y$ for a
  fresh $y$, and $\tmzero \closbc\rel \tmone$. Moreover, suppose $\tmzero
  \evalcbv \inctx{\ctxzero'}{\shift \vark {\tmzero'}}$. Then by the induction hypothesis, we
  know that there exist $\ctxone'$, $\tmone'$ such that $\tmone \evalcbv \inctx
  {\ctxone'}{\shift \vark {\tmone'}}$, and $\reset {\subst {\tmzero'} \vark
    {\lam \varx {\reset {\app {\vark'}{\inctx {\ctxzero'} \varx}}}}} \closbc\rel
  \reset {\subst{\tmone'} \vark {\lam \varx {\reset{\app {\vark'}{\inctx
            {\ctxone'} \varx}}}}}$ (*) for a fresh $\vark'$. Hence, we have
  $\inctx \ctxzero \tmzero \evalcbv \inctx \ctxzero{\inctx{\ctxzero'}{\shift
      \vark {\tmzero'}}}$ and $\tmone \evalcbv \inctx \ctxone {\inctx
    {\ctxone'}{\shift \vark {\tmone'}}}$, and we want to prove that $\reset
  {\subst {\tmzero'} \vark {\lam \varx {\reset {\app {\vark'}{\inctx \ctxzero
            {\inctx {\ctxzero'} \varx}}}}}} \closbc\rel \reset {\subst{\tmone'}
    \vark {\lam \varx {\reset{\app {\vark'}{\inctx \ctxone{\inctx {\ctxone'}
              \varx}}}}}}$ holds. Because $\inctx \ctxzero y \closc\rel \inctx
  \ctxone y$, we have $\lam y {\app {\vark'}{\inctx \ctxzero y}} \nf{\closc\rel}
  \lam y {\app {\vark'}{\inctx \ctxone y}}$ (**). Using (*) and (**), we have
  $\reset {\subst {\tmzero'} \vark {\lam \varx {\reset {\app {\lamp y {\app
              {\vark'}{\inctx \ctxzero y}}}{\inctx {\ctxzero'} \varx}}}}}
  \closbc\rel \reset {\subst{\tmone'} \vark {\lam \varx {\reset{\app {\lamp y
            {\app {\vark'}{\inctx \ctxone y}}}{\inctx {\ctxone'} \varx}}}}}$ by
  substitutivity of $\closbc\rel$. By Lemma \ref{l:rbisim-trick}, we have
  \begin{align*}
    \reset {\subst {\tmzero'} \vark {\lam \varx {\reset {\app {\vark'}{\inctx
              \ctxzero {\inctx {\ctxzero'} \varx}}}}}} & \rbisim \reset {\subst
      {\tmzero'} \vark {\lam \varx {\reset {\app {\lamp y {\app {\vark'}{\inctx
                  \ctxzero
                  y}}}{\inctx {\ctxzero'} \varx}}}}} \\
    \reset {\subst{\tmone'} \vark {\lam \varx {\reset{\app {\vark'}{\inctx
              \ctxone{\inctx {\ctxone'} \varx}}}}}} & \rbisim \reset
    {\subst{\tmone'} \vark {\lam \varx {\reset{\app {\lamp y {\app
                {\vark'}{\inctx \ctxone y}}}{\inctx {\ctxone'} \varx}}}}},
  \end{align*}
  which means that $\reset {\subst {\tmzero'} \vark {\lam \varx {\reset {\app
          {\vark'}{\inctx \ctxzero {\inctx {\ctxzero'} \varx}}}}}} \rbisim
  \closbc\rel \rbisim \reset {\subst{\tmone'} \vark {\lam \varx {\reset{\app
          {\vark'}{\inctx \ctxone{\inctx {\ctxone'} \varx}}}}}}$ holds. We have
  then the required result because $\rbisim \closbc\rel \rbisim \subseteq
  \closbc\rel$.  

  Assume $\inctx \rctxzero {\reset {\inctx \ctxzero \tmzero}} \closbc\rel \inctx
  \rctxone {\reset {\inctx \ctxone \tmone}}$ with $\inctx \rctxzero \varx
  \closbc\rel \inctx \rctxone \varx$, $\inctx \ctxzero \varx \closbc\rel \inctx
  \ctxone \varx$ ($\varx$ fresh), and $\tmzero \closbc\rel \tmone$. Moreover,
  suppose $\tmzero \evalcbv \valzero'$, $\inctx \ctxzero {\valzero'} \evalcbv
  \inctx {\ctxzero'}{\shift \vark {\tmzero'}}$, and $\inctx \rctxzero
  {\reset{\inctx {\ctxzero'}{\shift \vark {\tmzero'}}}} \evalcbv \valzero$. By
  the induction hypothesis, there exists $\valone'$ such that $\tmone \evalcbv \valone'$ and
  $\valzero' \rnf{\closbc\rel} \valone'$. We have $\inctx \ctxone \tmone \clocbv
  \inctx \ctxone {\valone'}$ and $\inctx \ctxzero {\valzero'} \closbc\rel \inctx
  \ctxone {\valone'}$. Because $\inctx \rctxzero {\reset{\inctx
      {\ctxzero'}{\shift \vark {\tmzero'}}}} \evalcbv \valzero$ takes at least
  one step (corresponding to the capture of $\ctxzero'$ by shift), we know that
  the evaluation $\inctx \ctxzero {\valzero'} \evalcbv \inctx {\ctxzero'}{\shift
    \vark {\tmzero'}}$ is in $m-1$ steps or less. Therefore, by the induction hypothesis, there
  exists $\ctxone'$, $\tmone'$ such that $\inctx \ctxone {\valone'} \evalcbv
  \inctx {\ctxone'}{\shift \vark {\tmone'}}$ and $\reset {\subst {\tmzero'}
    \vark {\lam y {\reset {\app {\vark'}{\inctx {\ctxzero'} y}}}}} \closbc\rel
  \reset{\subst {\tmone'} \vark {\lam y {\reset {\app {\vark'}{\inctx {\ctxone'}
            y}}}}}$ for fresh $y$ and $\vark'$. Because $\lam z z
  \rnf{\closbc\rel} \lam z z$ and $\closbc\rel$ is substitutive, we have $\reset
  {\subst {\tmzero'} \vark {\lam y {\reset {\app {\lamp z z}{\inctx {\ctxzero'}
            y}}}}} \closbc\rel \reset{\subst {\tmone'} \vark {\lam y {\reset
        {\app {\lamp z z}{\inctx {\ctxone'} y}}}}}$. Using
  Lemma~\ref{l:rbisim-trick}, we obtain
  \begin{align*}
    \reset {\subst {\tmzero'} \vark {\lam y {\reset {\app {\lamp z z}{\inctx
              {\ctxzero'} y}}}}} & \rbisim \reset {\subst {\tmzero'} \vark {\lam
        y {\reset {\inctx {\ctxzero'} y}}}} \\
    \reset{\subst {\tmone'} \vark {\lam y {\reset {\app {\lamp z z}{\inctx
              {\ctxone'} y}}}}} & \rbisim \reset{\subst {\tmone'} \vark {\lam y
        {\reset {\inctx {\ctxone'} y}}}}.
  \end{align*}
  Consequently, we have $\reset {\subst {\tmzero'} \vark {\lam y
      {\reset {\inctx {\ctxzero'} y}}}} \rbisim \closbc\rel \rbisim
  \reset{\subst {\tmone'} \vark {\lam y {\reset {\inctx {\ctxone'}
          y}}}}$, and because $\rbisim \closbc\rel \rbisim \subseteq
  \closbc \rel$, we have $\reset {\subst {\tmzero'} \vark {\lam y
      {\reset {\inctx {\ctxzero'} y}}}} \closbc\rel \reset{\subst
    {\tmone'} \vark {\lam y {\reset {\inctx {\ctxone'} y}}}}$. From
  here, the proof is the same as in the corresponding case of the
  proof of Lemma \ref{l:closc-bisim}.  \qed
\end{proof}

\end{document}